\documentclass[oribibl, 11pt]{llncs}
\usepackage{fullpage}

\usepackage{amsmath}
\usepackage{amssymb}
\usepackage{graphicx}
\usepackage{subfigure}
\usepackage{algorithm}
\usepackage{algpseudocode}
\usepackage{color}
\usepackage{hyperref}

\usepackage{fancyhdr}
\makeatletter
\def\blfootnote{\xdef\@thefnmark{}\@footnotetext}
\makeatother

\newcommand{\removed}[1]{}

\newcommand{\rem}{\mathbf}

\newcommand{\cs}{\mathcal{S}}

\newcommand{\ca}{\mathcal{A}}

\newcommand{\ra}{\rightarrow}
\newcommand{\rsa}{\rightsquigarrow}

\newcommand{\bs}{\backslash}

\newcommand{\dprime}{{\prime\prime}}

\renewcommand{\P}{\textup{P}}
\newcommand{\E}{\textup{E}}

\title{Simple and Efficient Local Codes for Distributed Stable Network Construction\thanks{Supported in part by the project ``Foundations of Dynamic Distributed Computing Systems'' (\textsf{FOCUS}) which is implemented under the ``ARISTEIA'' Action of the  Operational Programme ``Education and Lifelong Learning'' and is co-funded by the European Union (European Social Fund) and Greek National Resources.}}

\author{Othon Michail\inst{1} \and Paul G. Spirakis\inst{1,2}}
\institute{Computer Technology Institute \& Press ``Diophantus'' (CTI), Patras, Greece \and Department of Computer Science, University of Liverpool, UK\\
Email:\email{ michailo@cti.gr, P.Spirakis@liverpool.ac.uk}}

\begin{document}

\maketitle

\begin{abstract}
In this work, we study protocols (i.e. distributed algorithms) so that populations of distributed processes can \emph{construct networks}. In order to highlight the basic principles of distributed network construction we keep the model minimal in all respects. In particular, we assume \emph{finite-state processes} that all begin from the same initial state and all execute the same protocol (i.e. the system is \emph{homogeneous}). Moreover, we assume \emph{pairwise interactions} between the processes that are scheduled by an adversary. The only constraint on the \emph{adversary scheduler} is that it must be \emph{fair}, intuitively meaning that it must assign to every reachable configuration of the system a non-zero probability to occur. In order to allow processes to construct networks, we let them \emph{activate} and \emph{deactivate} their pairwise connections. When two processes interact, the protocol takes as input the states of the processes and the state of their connection and updates all of them. In particular, in every interaction, the protocol may activate an inactive connection, deactivate an active one, or leave the state of a connection unchanged. Initially all connections are inactive and the goal is for the processes, after interacting and activating/deactivating connections for a while, to end up with a desired \emph{stable network} (i.e. one that does not change any more). We give protocols (optimal in some cases) and lower bounds for several basic network construction problems such as \emph{spanning line}, \emph{spanning ring}, \emph{spanning star}, and \emph{regular network}. We provide proofs of correctness for all of our protocols and analyze the \emph{expected time to convergence} of most of them under a \emph{uniform random scheduler} that selects the next pair of interacting processes uniformly at random from all such pairs. Finally, we prove several \emph{universality} results by presenting generic protocols that are capable of simulating a Turing Machine (TM) and exploiting it in order to construct a large class of networks. Our universality protocols use a subset of the population (\emph{waste}) in order to distributedly construct there a TM able to decide a graph class in some given space. Then, the protocols repeatedly construct in the rest of the population (\emph{useful space}) a graph equiprobably drawn from all possible graphs. The TM works on this and accepts if the presented graph is in the class. We additionally show how to partition the population into $k$ \emph{supernodes}, each being a line of $\log k$ nodes, for the largest such $k$. This amount of local memory is sufficient for the supernodes to obtain unique names and exploit their names and their memory to realize nontrivial constructions. Delicate composition and reinitialization issues have to be solved for these general constructions to work.
\end{abstract}


\section{Introduction}
\label{sec:intro}

\subsection{Motivation}

Suppose a set of tiny computational devices (possibly at the nanoscale) is injected into a human circulatory system for the purpose of monitoring or even treating a disease. The devices are incapable of controlling their mobility. The mobility of the devices, and consequently the interactions between them, stems solely from the dynamicity of the environment, the blood flow inside the circulatory system in this case. Additionally, each device alone is incapable of performing any useful computation, as the small scale of the device highly constrains its computational capabilities. The goal is for the devices to accomplish their task via cooperation. To this end, the devices are equipped with a mechanism that allows them to create bonds with other devices (mimicking nature's ability to do so). So, whenever two devices come sufficiently close to each other and interact, apart from updating their local states, they may also become connected by establishing a physical connection between them. Moreover, two connected devices may at some point choose to drop their connection. In this manner, the devices can organize themselves into a desired global structure. This network-constructing self-assembly capability allows the artificial population of devices to evolve greater complexity, better storage capacity, and to adapt and optimize its performance to the needs of the specific task to be accomplished.

\subsection{Our Approach}

In this work, we study the fundamental problem of \emph{network construction} by a distributed computing system. The system consists of a set of processes that are capable of performing local computation (via pairwise interactions) and of forming and deleting connections between them. Connections between processes can be either \emph{physical} or \emph{virtual} depending on the application. In the most general case, a connection between two processes can be in one of a finite number of possible states. For example, state 0 could mean that the connection does not exist while state $i\in\{1,2,\ldots,k\}$, for some finite $k$, that the connection exists and has strength $i$. We consider here the simplest case, which we call the \emph{on/off} case, in which, at any time, a connection can either exist or not exist, that is there are just two states for the connections. If a connection exists we also say that it is \emph{active} and if it does not exist we say that it is \emph{inactive}. Initially all connections are inactive and the goal is for the processes, after interacting and activating/deactivating connections for a while, to end up with a desired \emph{stable network}. In the simplest case, the output-network is the one induced by the active connections and it is stable when no connection changes state any more.

Our aim in this work is to initiate this study by proposing and studying a very \emph{simple}, yet sufficiently generic, model for distributed network construction. To this end, we assume the computationally weakest type of processes. In particular, the processes are finite automata that all begin from the same initial state and all execute the same finite program which is stored in their memory (i.e. the system is \emph{homogeneous}). The communication model that we consider is also very minimal. In particular, we consider processes that are inhabitants of an \emph{adversarial environment} that has total control over the inter-process interactions. We model such an environment by an adversary scheduler that operates in discrete steps selecting in every step a pair of processes which then interact according to the common program. This represents very well systems of (not necessarily computational) entities that interact in pairs whenever two of them come sufficiently close to each other. When two processes interact, the program takes as input the states of the interacting processes and the state of their connection and outputs a new state for each process and a new state for the connection. The only restriction that we impose on the scheduler in order to study the constructive power of the model is that it is \emph{fair}, by which we mean the weak requirement that, at every step, it assigns to every reachable configuration of the system a non-zero probability to occur. In other words, a fair scheduler cannot forever conceal an always reachable configuration of the system. Note that such a generic scheduler gives no information about the running time of our constructors. Thus, to estimate the efficiency of our solutions we assume a \emph{uniform random scheduler}, one of the simplest fair probabilistic schedulers. The uniform random scheduler selects in every step independently and uniformly at random a pair of processes to interact from all such pairs. What renders this model interesting is its ability to achieve complex global behavior via a set of notably simple, uniform (i.e. with codes that are independent of the size of the system), homogeneous, and cooperative entities.

We now give a simple illustration of the above. Assume a set of $n$ very weak processes that can only be in one of two states, ``black'' or ``red''. Initially, all processes are black. We can think of the processes as small particles that move randomly in a fair solution. The particles are capable of forming and deleting physical connections between them, by which we mean that, whenever two particles interact, they can read and write the state of their connection. Moreover, for simplicity of the model, we assume that fairness of the solution is independent of the states of the connections. This is in contrast to schedulers that would take into account the geometry of the active connections and would, for example, forbid two non-neighboring particles of the same component to interact with each other. In particular, we assume that throughout the execution every pair of processes may be selected for interaction. Consider now the following simple problem. We want to identically program the initially disorganized particles so that they become self-organized into a \emph{spanning star}. In particular, we want to end up with a unique black particle  connected (via active connections) to $n-1$ red particles and all other connections (between red particles) being inactive. Equivalently, given a (possibly physical) system that tends to form a spanning star we would like to unveil the code behind this behavior. Consider the following program. When two black particles that are not connected interact, they become connected and one of them becomes red. When two connected red particles interact they become disconnected (i.e. reds repel). Finally, when a black and a red that are not connected interact they become connected (i.e. blacks and reds attract). The protocol forms a spanning star as follows. As whenever two blacks interact only one survives and the other becomes red, eventually a unique black will remain and all other particles will be red (we say ``eventually'', meaning ``in finite time'', because we do not know how much time it will take for all blacks to meet each other but from fairness we know that this has to occur in a finite number of steps). As blacks and reds attract while reds repel, it is clear that eventually the unique black will be connected to all reds while every pair of reds will be disconnected. Moreover, no rule of the program can modify such a configuration thus the constructed spanning star is stable (see Figure \ref{fig:global-star}). It is worth noting that this very simple protocol is optimal both w.r.t. to the number of states that it uses and w.r.t. to the time it takes to construct a stable spanning star under the uniform random scheduler.

\begin{figure}[!hbtp]
   \centering{
        \subfigure[]{
        \includegraphics[width=0.27\textwidth]{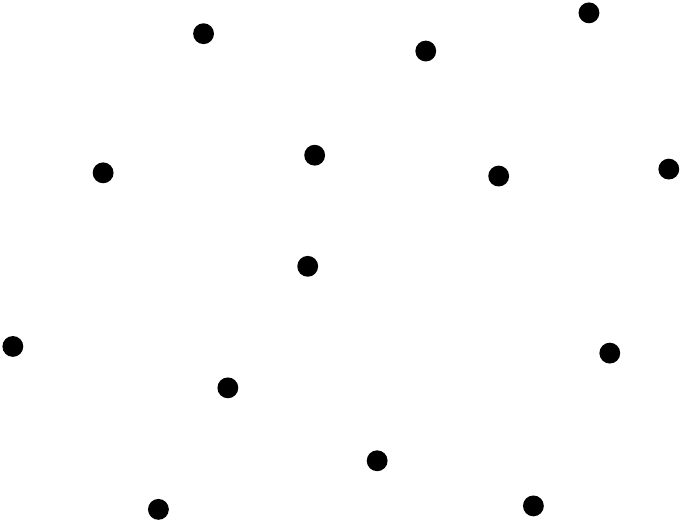}
        \label{fig:gs1}}
	\hspace{1cm}
        \subfigure[]{
        \includegraphics[width=0.28\textwidth]{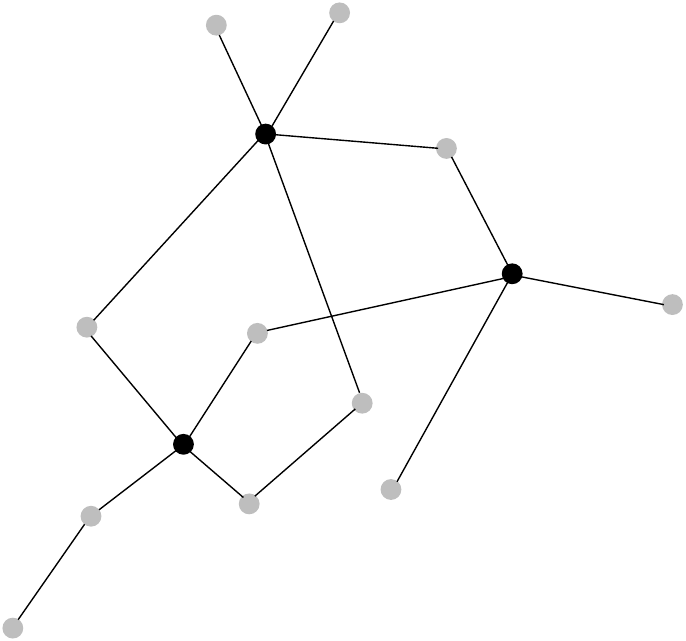}
        \label{fig:gs2}}
	\hspace{1cm}
        \subfigure[]{
        \includegraphics[width=0.23\textwidth]{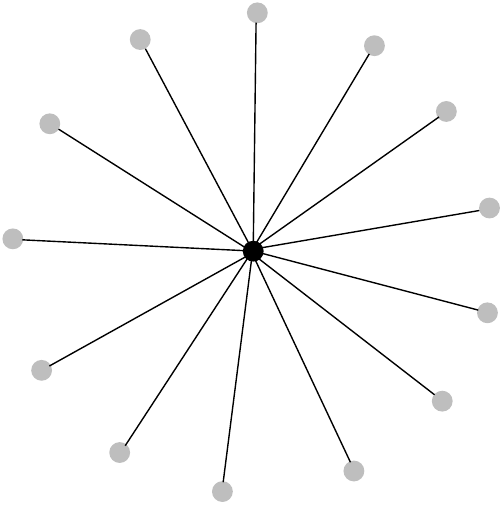}
        \label{fig:gs3}}
        }
   \caption{(a) Initially all particles are black and no active connections exist. (b) After a while, only 3 black particles have survived each having a set of red neighbors (red particles appear as gray here). Note that some red particles are also connected to red particles. The tendency is for the red particles to repel red particles and attract black particles. (c) A unique black has survived, it has attracted all red particles, and all connections between red particles have been deactivated. The construction is a stable spanning star.} \label{fig:global-star}
\end{figure}

Our model for network construction is strongly inspired by the Population Protocol model \cite{AADFP06} and the Mediated Population Protocol model \cite{MCS11-2}. In the former, connections do not have states. States on the connections were first introduced in the latter. The main difference to our model is that \emph{in those models the focus was on the computation of functions of some input values and not on network construction}. Another important difference is that we allow the edges to choose between \emph{only two possible states} which was not the case in \cite{MCS11-2}. Interestingly, when operating under a uniform random scheduler, population protocols are formally equivalent to \emph{chemical reaction networks} (CRNs) which model chemistry in a \emph{well-mixed solution} \cite{Do14}. CRNs are widely used to describe information processing occurring in natural cellular regulatory networks, and with upcoming advances in synthetic biology, CRNs are a promising programming language for the design of artificial molecular control circuitry. However, CRNs and population protocols can only capture the dynamics of molecular counts and not of structure formation. Our model then may also be viewed as an extension of population protocols and CRNs aiming to capture the stable structures that may occur in a well-mixed solution. From this perspective, our goal is to determine what stable structures can result in such systems (natural or artificial), how fast, and under what conditions (e.g. by what underlying codes/reaction-rules). Most computability issues in the area of population protocols have now been resolved. Finite-state processes on a complete interaction network, i.e. one in which every pair of processes may interact, (and several variations) compute the \emph{semilinear predicates} \cite{AAER07}. Semilinearity persists up to $o(\log\log n)$ local space but not more than this \cite{MNPS11}.  If additionally the connections between processes can hold a state from a finite domain (note that this is a stronger requirement than the on/off that the present work assumes) then the computational power dramatically increases to the commutative subclass of $\rem{NSPACE}(n^2)$ \cite{MCS11-2}. Other important works include \cite{GR09} which equipped the nodes of population protocols with unique ids and \cite{BBCK10} which introduced a (weak) notion of speed of the nodes that allowed the design of fast converging protocols with only weak requirements. For a very recent introductory text see \cite{MCS11}.

The paper essentially consists of two parts. In the first part, we give simple (i.e. small) and efficient (i.e. polynomial-time) protocols for the construction of several fundamental networks. In particular, we give protocols for spanning lines, spanning rings, cycle-covers, partitioning into cliques, and regular networks (formal definitions of all problems considered can be found in Section \ref{sec:problems}). We remark that the spanning line problem is of outstanding importance because it constitutes a basic ingredient of universal constructors. We give three different protocols for this problem each improving on the running time but using more states to this end. Additionally, we establish a $\Omega(n\log n)$ generic lower bound on the expected running time of all constructors that construct a spanning network and a $\Omega(n^2)$ lower bound for the spanning line, where $n$ throughout this work denotes the number of processes. Our fastest protocol for the problem runs in $O(n^3)$ expected time and uses 9 states while our simplest uses only 5 states but pays in an expected time which is between $\Omega(n^4)$ and $O(n^5)$. In the second part, we investigate the more generic question of \emph{what is in principle constructible by our model}. We arrive there at several satisfactory characterizations establishing some sort of universality of the model. The main idea is as follows. To construct a decidable graph-language $L$ we (i) construct on $k$ of the processes (called the \emph{waste}) a network $G_1$ capable of simulating a Turing Machine (abbreviated ``TM'' throughout the paper) and of constructing a random network on the remaining $n-k$ processes (called the \emph{useful space}), (ii) use $G_1$ to construct a random network $G_2\in G_{n-k,1/2}$ on the remaining $n-k$ processes, (iii) execute on $G_1$ the TM that decides $L$ with $G_2$ as input. If the TM accepts, then we output $G_2$ (note that this is not a terminating step - the reason why will become clear in Section \ref{sec:gencon}; the protocol just freezes and its output forever remains $G_2$), otherwise we go back to (ii) and repeat. Using this core idea we prove several universality results for our model. Additionally, we show how to organize the population into a distributed system with names and logarithmic local memories.

In Section \ref{sec:rw}, we discuss further related literature. Section \ref{sec:prel} brings together all definitions and basic facts that are used throughout the paper. In particular, in Section \ref{sec:model} we formally define the model of network constructors, Section \ref{sec:problems} formally defines all network construction problems that are considered in this work, and in Section \ref{sec:basic-processes} we identify and analyze a set of basic probabilistic processes that are recurrent in the analysis of the running times of network constructors. In Section \ref{sec:global-line}, we study the spanning line problem. In Section \ref{sec:basic-con}, we provide direct constructors for all the other basic network construction problems. Section \ref{sec:gencon} presents our \emph{universality} results. Finally, in Section \ref{sec:conclusions} we conclude and give further research directions that are opened by our work.

\section{Further Related Work}
\label{sec:rw}

\noindent\textbf{Algorithmic Self-Assembly.} There are already several models trying to capture the self-assembly capability of natural processes with the purpose of engineering systems and developing algorithms inspired by such processes. For example, \cite{Do12} proposes to learn how to program molecules to manipulate themselves, grow into machines and at the same time control their own growth. The research area of ``algorithmic self-assembly'' belongs to the field of ``molecular computing''. The latter was initiated by Adleman \cite{Ad94}, who designed interacting DNA molecules to solve an instance of the Hamiltonian path problem. The model guiding the study in algorithmic self-assembly is the Abstract Tile Assembly Model (aTAM) \cite{Wi98,RW00} and variations. In contrast to those models that try to incorporate the exact molecular mechanisms (like e.g. temperature, energy, and bounded degree), we propose a very abstract combinatorial rule-based model, free of specific application-driven assumptions, with the aim of revealing the fundamental laws governing the distributed (algorithmic) generation of networks. Our model may serve as a common substructure to more applied models (like assembly models or models with geometry restrictions) that may be obtained from our model by imposing restrictions on the scheduler, the degree, and the number of local states (see Section \ref{sec:conclusions} for several interesting variations of our model).\\

\noindent\textbf{Distributed Network Construction.} To the best of our knowledge, classical distributed computing has not considered the problem of constructing an actual communication network from scratch. From the seminal work of Angluin \cite{An80} that initiated the theoretical study of distributed computing systems up to now, the focus has been more on assuming a given communication topology and constructing a virtual network over it, e.g. a spanning tree for the purpose of fast dissemination of information. Moreover, these models assume most of the time unique identities, unbounded memories, and message-passing communication. Additionally, a process always communicates with its neighboring processes (see \cite{Ly96} for all the details). An exception is the area of geometric pattern formation by mobile robots (cf. \cite{SY99,DFSY10} and references therein). A great difference, though, to our model is that in mobile robotics the computational entities have complete control over their mobility and thus over their future interactions. That is, the goal of a protocol is to result in a desired interaction pattern while in our model the goal of a protocol is to construct a network while operating under a totally unpredictable interaction pattern. Very recently, a model inspired by the behavior of ameba that allows algorithmic research on self-organizing particle systems was proposed \cite{DGRS13}. The goal is for the particles to self-organize in order to adapt to a desired shape without any central control, which is quite similar to our objective, however the two models seem two have little in common. In the same work, the authors observe that, in contrast to the considerable work that has been performed w.r.t. to systems (e.g. self-reconfigurable robotic systems), only very little theoretical work has been done in this area. This further supports the importance of introducing a simple yet sufficiently generic model for distributed network construction, as we do in this work.\\

\noindent\textbf{Cellular Automata.} A cellular automaton (cf. e.g. \cite{Sc11}) consists of a grid of cells each cell being a finite automaton. A cell updates its own state by reading the states of its neighboring cells (e.g. 2 in the 1-dimensional case and 4 in the 2-dimensional case). All cells may perform the updates in discrete synchronous steps or updates may occur asynchronously. Cellular automata have been used as models for self-replication, for modeling several physical systems (e.g. neural activity, bacterial growth, pattern formation in nature), and for understanding emergence, complexity, and self-organization issues. Though there are some similarities there are also significant differences between our model and cellular automata. One is that in our model the interaction pattern is nondeterministic as it depends on the scheduler and a process may interact with any other process of the system and not just with some predefined neighbors. Moreover, our model has a direct capability of forming networks whereas cellular automata can form networks only indirectly (an edge between two cells $u$ and $v$ has to be represented as a line of cells beginning at $u$, ending at $v$ and all cells on the line being in a special edge-state). In fact, cellular automata are more suitable for studying the formation of patterns on e.g. a discrete surface of static cells while our model is more suitable for studying how a totally dynamic (e.g. mobile) and initially disordered collection of entities can self-organize into a network.\\

\noindent\textbf{Social Networks.} There is a great amount of work dealing with networks formed by a group of interacting individuals. Individuals, also called players, which may, for example, be people, animals, or companies, depending on the application, usually have incentives and connections between individuals indicate some social relationship, like for example friendship. The network is formed by allowing the individuals to form or delete connections usually selfishly by trying to maximize their own utility. The usual goal there is to study how the whole network affects the outcome of a specific interaction, to predict the network that will be formed by a set of selfish individuals, and to characterize the quality of the network formed (e.g. its efficiency). See e.g. \cite{Ja05,BEK13}. This is a game-theoretic setting which is very different from the setting considered here as the latter does not include incentives and utilities. Another important line of research considers random social networks in which new links are formed according to some probability distribution. For example, in \cite{BA99} it was shown that growth and preferential attachment that characterize a great majority of social networks (like, for example, the Internet) results in scale-free properties that are not predicted by the Erd\"{o}s-R\'{e}nyi random graph model \cite{ER59,Bo01}. Though, in principle, we allow processes to perform a coin tossing during an interaction, our focus is not on the formation of a random network but on cooperative (algorithmic) construction according to a common set of rules. In summary, our model looks more like a standard dynamic distributed computing system in which the interacting entities are computing processes that all execute the same program.\\

\noindent\textbf{Network Formation in Nature.} Nature has an intrinsic ability to form complex structures and networks via a process known as \emph{self-assembly}. By self-assembly, small components (like e.g. molecules) automatically assemble into large, and usually complex structures (like e.g. a crystal). There is an abundance of such examples in the physical world. Lipid molecules form a cell's membrane, ribosomal proteins and RNA coalesce into functional ribosomes, and bacteriophage virus proteins self-assemble a capsid that allows the virus to invade bacteria \cite{Do12}. Mixtures of RNA fragments that self-assemble into self-replicating ribozymes spontaneously form cooperative catalytic cycles and networks. Such cooperative networks grow faster than selfish autocatalytic cycles indicating an intrinsic ability of RNA populations to evolve greater complexity through cooperation \cite{VMC12}. Through billions of years of prebiotic molecular selection and evolution, nature has produced a basic set of molecules. By combining these simple elements, natural processes are capable of fashioning an enormously diverse range of fabrication units, which can further self-organize into refined structures, materials and molecular machines that not only have high precision, flexibility and error-correction capacity, but are also self-sustaining and evolving. In fact, nature shows a strong preference for bottom-up design.

Systems and solutions inspired by nature have often turned out to be extremely practical and efficient. For example, the bottom-up approach of nature inspires the fabrication of biomaterials by attempting to mimic these phenomena with the aim of creating new and varied structures with novel utilities well beyond the gifts of nature \cite{Zh03}. Moreover, there is already a remarkable amount of work envisioning our future ability to engineer computing and robotic systems by manipulating molecules with nanoscale precision. Ambitious long-term applications include molecular computers \cite{BPSPF10} and miniature (nano)robots for surgical instrumentation, diagnosis and drug delivery in medical applications (e.g. it has very recently been reported that DNA nanorobots could even kill cancer cells \cite{DBC12}) and monitoring in extreme conditions (e.g. in toxic environments). However, the road towards this vision passes first through our ability to discover \emph{the laws governing the capability of distributed systems to construct networks}. The gain of developing such a theory will be twofold: It will give some insight to the role (and the mechanisms) of network formation in the complexity of natural processes and it will allow us engineer artificial systems that achieve this complexity.

\section{Preliminaries}
\label{sec:prel}

\subsection{A Model of Network Constructors}
\label{sec:model}

\begin{definition}
A \emph{Network Constructor} (NET) is a distributed protocol defined by a 4-tuple $(Q,q_0,Q_{out},\delta)$, where
$Q$ is a finite set of \emph{node-states}, $q_0\in Q$ is the \emph{initial node-state}, $Q_{out}\subseteq Q$ is the set of \emph{output node-states}, and $\delta : Q\times Q\times \{0,1\} \rightarrow Q\times Q\times \{0,1\}$ is the \emph{transition function}.
\end{definition}

If $\delta(a,b,c) = (a^{\prime},b^{\prime},c^{\prime})$, we call $(a,b,c) \rightarrow (a^{\prime},b^{\prime},c^{\prime})$ a \emph{transition} (or \emph{rule}) and we define $\delta_{1}(a,b,c) = a^{\prime}$, $\delta_{2}(a,b,c) = b^{\prime}$, and $\delta_{3}(a,b,c) = c^{\prime}$. A transition $(a,b,c) \rightarrow (a^{\prime},b^{\prime},c^{\prime})$ is called \emph{effective} if $x\neq x^\prime$ for at least one $x\in\{a,b,c\}$ and \emph{ineffective} otherwise. When we present the transition function of a protocol we only present the effective transitions. Additionally, we agree that the \emph{size} of a protocol is the number of its states, i.e. $|Q|$.

The system consists of a population $V_I$ of $n$ distributed \emph{processes} (also called \emph{nodes} when clear from context). In the generic case, there is an underlying \emph{interaction graph} $G_I=(V_I,E_I)$ specifying the permissible interactions between the nodes. Interactions in this model are always pairwise. In this work, $G_I$ is a \emph{complete undirected interaction graph}, i.e. $E_I=\{uv:u,v\in V_I \text{ and } u\neq v\}$, where $uv=\{u,v\}$. Initially, all nodes in $V_I$ are in the initial node-state $q_0$. 

A central assumption of the model is that edges have binary states. An edge in state 0 is said to be \emph{inactive} while an edge in state 1 is said to be \emph{active}. All edges are initially inactive.

Execution of the protocol proceeds in discrete steps. In every step, a pair of nodes $uv$ from $E_I$ is selected by an \emph{adversary scheduler} and these nodes interact and update their states and the state of the edge joining them according to the transition function $\delta$. In particular, we assume that, for all distinct node-states $a,b\in Q$ and for all edge-states $c\in\{0,1\}$, $\delta$ specifies either $(a,b,c)$ or $(b,a,c)$. So, if $a$, $b$, and $c$ are the states of nodes $u$, $v$, and edge $uv$, respectively, then the unique rule corresponding to these states, let it be $(a,b,c)\rightarrow (a^\prime,b^\prime,c^\prime)$, is applied, the edge that was in state $c$ updates its state to $c^\prime$ and if $a\neq b$, then $u$ updates its state to $a^\prime$ and $v$ updates its state to $b^\prime$, if $a=b$ and $a^\prime=b^\prime$, then both nodes update their states to $a^\prime$, and if $a=b$ and $a^\prime\neq b^\prime$, then the node that gets $a^\prime$ is drawn equiprobably from the two interacting nodes and the other node gets $b^\prime$. 

A \emph{configuration} is a mapping $C : V_I\cup E_I \rightarrow Q\cup \{0,1\}$ specifying the state of each node and each edge of the interaction graph. Let $C$ and $C^{\prime}$ be configurations, and let $u$, $\upsilon$ be distinct nodes. We say that \emph{$C$ goes to $C^{\prime}$ via encounter $e=u\upsilon$}, denoted $C \stackrel{e}\rightarrow C^{\prime}$, if $(C^{\prime}(u),C^{\prime}(v),C^{\prime}(e))=$ $\delta(C(u),C(v),C(e)) \text{ or }$ $(C^{\prime}(v),C^{\prime}(u),C^{\prime}(e))=$ $\delta(C(v),C(u),C(e)) \text{ and }$ $C^{\prime}(z)=$ $C(z), \mbox{ for all } z\in (V_I\bs\{u,v\})\cup (E_I\bs\{e\})$. We say that \emph{$C^\prime$ is reachable in one step from $C$}, denoted $C\rightarrow C^{\prime}$, if $C \stackrel{e}\rightarrow C^{\prime}$ for some encounter $e\in E_I$. We say that $C^{\prime}$ is \emph{reachable} from $C$ and write $C\rsa C^{\prime}$, if there is a sequence of configurations $C=C_{0},C_{1},\ldots,C_{t}=C^{\prime}$, such that $C_{i}\rightarrow C_{i+1}$ for all $i$, $0\leq i <t$.

An \emph{execution} is a finite or infinite sequence of configurations $C_{0},C_{1},$ $C_{2},\ldots$, where $C_{0}$ is an initial configuration and $C_{i}\rightarrow C_{i+1}$, for all $i\geq 0$. A \emph{fairness condition} is imposed on the adversary to ensure the protocol makes progress. An infinite execution is \emph{fair} if for every pair of configurations $C$ and $C^{\prime}$ such that $C\rightarrow C^{\prime}$, if $C$ occurs infinitely often in the execution then so does $C^{\prime}$. In what follows, every execution of a NET will by definition considered to be fair.

We define the \emph{output of a configuration} $C$ as the graph $G(C)=(V,E)$ where $V=\{u\in V_I: C(u)\in Q_{out}\}$ and $E=\{uv:u,v\in V,$ $u\neq v$, and $C(uv)=1\}$. In words, the output-graph of a configuration consists of those nodes that are in output states and those edges between them that are active, i.e. the active subgraph induced by the nodes that are in output states. The output of an execution $C_0,C_1,\ldots$ is said to \emph{stabilize} (or \emph{converge}) to a graph $G$ if there exists some step $t\geq 0$ s.t. $G(C_i)=G$ for all $i\geq t$, i.e. from step $t$ and onwards the output-graph remains unchanged. Every such configuration $C_i$, for $i\geq t$, is called \emph{output-stable}. The \emph{running time} (or \emph{time to convergence}) of an execution is defined as the minimum such $t$ (or $\infty$ if no such $t$ exists). Throughout the paper, whenever we study the running time of a NET, we assume that interactions are chosen by a \emph{uniform random scheduler} which, in every step, selects independently and uniformly at random one of the $|E_I|=n(n-1)/2$ possible interactions. In this case, the running time becomes a random variable (abbreviated ``r.v.'') $X$ and our goal is to obtain bounds on the expectation $\E[X]$ of $X$. Note that the uniform random scheduler is fair with probability 1.

\begin{definition}
We say that an execution of a NET on $n$ processes \emph{constructs a graph} (or \emph{network}) $G$, if its output stabilizes to a graph isomorphic to $G$.
\end{definition}

\begin{definition}
We say that a NET $\ca$ \emph{constructs a graph language $L$ with useful space $g(n)\leq n$}, if $g(n)$ is the greatest function for which: (i) for all $n$, every execution of $\ca$ on $n$ processes constructs a $G\in L$ of order at least $g(n)$ (provided that such a $G$ exists) and, additionally, (ii) for all $G\in L$ there is an execution of $\ca$ on $n$ processes, for some $n$ satisfying $|V(G)|\geq g(n)$, that constructs $G$. Equivalently, we say that \emph{$\ca$ constructs $L$ with waste $n-g(n)$}.
\end{definition}

\begin{definition}
Define $\rem{REL}(g(n))$ to be the class of all graph languages that are constructible with useful space $g(n)$ by a NET. We call $\rem{REL}(\cdot)$ the \emph{relation} or \emph{on/off} class. 
\end{definition}

Also define $\rem{PREL}(g(n))$ in precisely the same way as $\rem{REL}(g(n))$ but in the extension of the above model in which every pair of processes is capable of tossing an unbiased coin during an interaction between them. In particular, in the weakest probabilistic version of the model, we allow transitions that with probability $1/2$ give one outcome and with probability $1/2$ another. Additionally, we require that all graphs have the same probability to be constructed by the protocol.

We denote by $\rem{DGS}(f(l))$ (for ``Deterministic Graph Space'') the class of all graph languages that are decidable by a TM of (\emph{binary}) space $f(l)$, where $l$ is the length of the adjacency matrix encoding of the input graph.

\subsection{Problem Definitions}
\label{sec:problems}

We here provide formal definitions of all the network construction problems that are considered in this work. Protocols and bounds for these problems are presented in Sections \ref{sec:global-line} and \ref{sec:basic-con}.\\

\noindent\textbf{Global line.} The goal is for the $n$ distributed processes to construct a spanning line, i.e. a connected graph in which 2 nodes have degree 1 and $n-2$ nodes have degree 2.\\

\noindent\textbf{Cycle cover.} Every process in $V_I$ must eventually have degree 2. The result is a collection of node-disjoint cycles spanning $V_I$.\\

\noindent\textbf{Global star.} The processes must construct a spanning star, i.e. a connected graph in which 1 node, called the \emph{center}, has degree $n-1$ and $n-1$ nodes, called the \emph{peripheral nodes}, have degree 1.\\

\noindent\textbf{Global ring.} The processes must construct a spanning ring, i.e. a connected graph in which every node has degree 2.\\

\noindent\textbf{$k$-regular connected.} The generalization of global ring in which every node has degree $k\geq 2$ (note that $k$ is a constant and a protocol for the problem must run correctly on any number $n$ of processes).\\

\noindent\textbf{$c$-cliques.} The processes must partition themselves into $\lfloor n/c\rfloor$ cliques of order $c$ each (again $c$ is a constant).\\

\noindent\textbf{Replication.} The protocol is given an input graph $G=(V,E)$ on a subset $V_1$ of the processes ($|V_1|=|V|$). The processes in $V_1$ are initially in state $q_0$ and the edges of $E$ are the active edges between them. All other edges in $E_I$ are initially inactive. The processes in $V_2=V_I\bs V_1$ are initially in state $r_0$. The goal is to create a \emph{replica} of $G$ on $V_2$, provided that $|V_2|\geq |V_1|$. Formally, we want, in every execution, the output induced by the active edges between the nodes of $V_2$ to stabilize to a graph isomorphic to $G$.

\subsection{Basic Probabilistic Processes}
\label{sec:basic-processes}
 
We now present a set of very fundamental probabilistic processes that are recurrent in the analysis of the running times of network constructors. All these processes assume a uniform random scheduler and are applications of the standard coupon collector problem. In most of these processes, we ignore the states of the edges and focus only on the dynamics of the node-states, that is we consider rules of the form $\delta:Q\times Q\rightarrow Q\times Q$. Throughout this section, we call a step a \emph{success} if an effective rule applies on the interacting nodes and we denote by $X$ the r.v. of the running time of the processes.\\

\noindent\textbf{One-way epidemic.} Consider the protocol in which the only effective transition is $(a,b)\ra (a,a)$. Initially, there is a single $a$ and $n-1$ $b$s and we want to estimate the expected number of steps until all nodes become $a$s.

\begin{proposition} \label{pro:one-way-ep}
The expected time to convergence of a one-way epidemic (under the uniform random scheduler) is $\Theta(n\log n)$.
\end{proposition}
\begin{proof}
Let $X$ be a r.v. defined to be the number of steps until all $n$ nodes are in state $a$. Call a step a success if an effective rule applies and a new $a$ appears on some node. Divide the steps of the protocol into \emph{epochs}, where epoch $i$ begins with the step following the $(i-1)$st success and ends with the step at which the $i$th success occurs. Let also the r.v. $X_i$, $1\leq i\leq n-1$ be the number of steps in the $i$-th epoch. Let $p_i$ be the probability of success at any step during the $i$-th epoch. We have $p_i=\frac{i(n-i)}{m}=\frac{2i(n-i)}{n(n-1)}$, where $m=|E_I|=n(n-1)/2$ denotes the total number of possible interactions and $\E[X_i]=1/p_i=\frac{n(n-1)}{2i(n-i)}$. By linearity of expectation we have 
\begin{align*}
\E[X]&=\E[\sum_{i=1}^{n-1} X_i]=\sum_{i=1}^{n-1}\E[X_i]=\sum_{i=1}^{n-1}\frac{n(n-1)}{2i(n-i)}=\frac{n(n-1)}{2}\sum_{i=1}^{n-1}\frac{1}{i(n-i)}\\
&=\frac{n(n-1)}{2}\sum_{i=1}^{n-1}\frac{1}{n}\left (\frac{1}{i}+\frac{1}{n-i}\right )=\frac{(n-1)}{2}\left [\sum_{i=1}^{n-1}\frac{1}{i}+\sum_{i=1}^{n-1} \frac{1}{n-i}\right ]\\
&=\frac{(n-1)}{2}2H_{n-1}= (n-1)[\ln (n-1)+\Theta(1)]\\
&=\Theta(n\log n),
\end{align*}
\qed
\end{proof}
where $H_n$ denotes the $n$th Harmonic number.\\

\noindent\textbf{One-to-one elimination.} All nodes are initially in state $a$. The only effective transition of the protocol is $(a,a)\ra (a,b)$. We are now interested in the expected time until a single $a$ remains. We call the process one-to-one elimination because $a$s are only eliminated with themselves. A straightforward application is in protocols that elect a unique leader by beginning with all nodes in the leader state and eliminating a leader whenever two leaders interact.

\begin{proposition} \label{pro:one-to-one}
The expected time to convergence of a one-to-one elimination is $\Theta(n^2)$.
\end{proposition}
\begin{proof}
Epoch $i$ begins with the step following the $i$th success and ends with the step at which the $(i+1)$st success occurs. The probability of success during the $i$th epoch, for $0\leq i\leq n-2$, is $p_i=[(n-i)(n-i-1)/2]/[n(n-1)/2]=[(n-i)(n-i-1)]/[n(n-1)]$ and
\begin{align*}
\E[X]&=n(n-1)\sum_{i=0}^{n-2} \frac{1}{(n-i)(n-i-1)}= n(n-1)\sum_{i=2}^{n} \frac{1}{i(i-1)}<n(n-1)\sum_{i=2}^{n} \frac{1}{(i-1)^2}\\
&=n(n-1)\sum_{i=1}^{n-1} \frac{1}{i^2}\leq n(n-1)(1-\frac{1}{n})<n^2.
\end{align*}
The above uses the fact that $\sum_{i=1}^{n-1} 1/i^2$ is bounded from above by $(1-1/n)$. This holds because $\sum_{i=1}^{n-1} 1/i^2\leq \int_{s=1}^{n} (1/s^2) \mathrm{d}s=\left[-s^{-1}\right]_{s=1}^{n}=1-1/n$.

Now, for the lower bound, observe that the last two $a$s need on average $n(n-1)/2$ steps to meet each other. As $n(n-1)/2\leq \E[X] < n^2$, we conclude that $\E[X]=\Theta(n^2)$.
\qed
\end{proof}

A slight variation of the one-to-one elimination protocol constructs a \emph{maximum matching}, i.e. a matching of cardinality $\lfloor n/2\rfloor$ (which is a perfect matching in case $n$ is even). The variation is $(a,a,0)\ra (b,b,1)$ and the running time of a one-to-one elimination, i.e. $\Theta(n^2)$, is an upper bound on this variation. For the lower bound, notice that when only two (or three) $a$s remain the expected number of steps for a success is $n(n-1)/2$ ($n(n-1)/6$, respectively), that is the running time is also $\Omega(n^2)$. We conclude that there is a protocol that constructs a maximum matching in an expected number of $\Theta(n^2)$ steps.\\

\noindent\textbf{One-to-all elimination.} All nodes are initially in state $a$. The effective rules of the protocol are $(a,a)\ra (b,a)$ and $(a,b)\ra (b,b)$. We are now interested in the expected time until no $a$ remains. The process is called one-to-all elimination because $a$s are eliminated not only when they interact with $a$s but also when they interact with $b$s. At a first sight, it seems to run faster than a one-way epidemic as $b$s still propagate towards $a$s as in a one-way epidemic but now $b$s are also created when two $a$s interact. We show that this is not the case.

\begin{proposition} \label{pro:one-to-all}
The expected time to convergence of a one-to-all elimination is $\Theta(n\log n)$.
\end{proposition}
\begin{proof} 
The probability of success during the $i$th epoch, for $0\leq i\leq n-1$, is $p_i=1-[i(i-1)/2]/[n(n-1)/2]=[n(n-1)-i(i-1)]/[n(n-1)]$ and
\begin{equation*}
\E[X]=n(n-1)\sum_{i=0}^{n-1} \frac{1}{n(n-1)-i(i-1)}.
\end{equation*}

For the upper bound, we have 
\begin{align*}
\E[X] &= n(n-1)\sum_{i=0}^{n-1} \frac{1}{n(n-1)-i(i-1)}
< n(n-1)\left[\sum_{i=0}^{n-2} \frac{1}{(n-1)^2-i^2}\right ] + \frac{n}{2}\\
&= \frac{n}{2}\left(\sum_{i=0}^{n-2}\frac{1}{n-i-1} + \sum_{i=0}^{n-2}\frac{1}{n+i-1} + 1 \right )
= \frac{n}{2}\left(\sum_{i=1}^{n-1}\frac{1}{i}+\sum_{i=1}^{2n-3}\frac{1}{i}-\sum_{i=1}^{n-2}\frac{1}{i}+1\right )\\
&= \frac{n}{2}\left(\frac{1}{n-1} + \sum_{i=1}^{2n-3}\frac{1}{i} + 1 \right )
= \frac{n}{2}H_{2n-3}+\frac{n}{2}+\frac{n}{2(n-1)}\\
&< n(H_{2n}+1)
= n[\ln 2n + \Theta(1)].
\end{align*}

For the lower bound, we have
\begin{align*}
\E[X] &= n(n-1)\sum_{i=0}^{n-1} \frac{1}{n(n-1)-i(i-1)}
> n(n-1)\sum_{i=0}^{n-1} \frac{1}{n^2-(i-1)^2}\\
&= \frac{n-1}{2}\left(\sum_{i=0}^{n-1} \frac{1}{n-i+1}+\sum_{i=0}^{n-1} \frac{1}{n+i-1}\right)\\
&= \frac{n-1}{2}\left(\sum_{i=1}^{n+1} \frac{1}{i}+\sum_{i=1}^{2n-2} \frac{1}{i} - \sum_{i=1}^{n-2} \frac{1}{i} - 1\right)\\
&= \frac{n-1}{2}\left(\sum_{i=1}^{2n-2} \frac{1}{i}+\frac{1}{n-1} + \frac{1}{n} + \frac{1}{n+1} - 1\right)
> \frac{n-1}{2}(H_{2n-2}-1)\\
&= \frac{n-1}{2}[\ln (2n-2) + \Theta(1)].
\end{align*}
We conclude that $\E[X]=\Theta(n\log n)$.
\qed
\end{proof}

\noindent\textbf{Meet everybody.} A single node $u$ is initially in state $a$ and all other nodes are in state $b$. The only effective transition is $(a,b)\ra (a,c)$. We study the time until all $b$s become $c$s which is equal to the time needed for $u$ to interact with every other node.

\begin{proposition} \label{pro:meet-ever}
The expected time to convergence of a meet everybody is $\Theta(n^2\log n)$.
\end{proposition}
\begin{proof}
Assume that in every step $u$ participates in an interaction. Then $u$ must collect the $n-1$ coupons which are $n-1$ different nodes that it must interact with. Clearly, in every step, every node has the same probability to interact with $u$, i.e. $1/(n-1)$, and this is the classical coupon collector problem that takes average time $\Theta(n\log n)$. But on average $u$ needs $\Theta(n)$ steps to participate in an interaction, thus the total time is $\Theta(n^2\log n)$. 
\qed
\end{proof} 

\noindent\textbf{Node cover.} All nodes are initially in state $a$. The only effective transitions are $(a,a)\ra (b,b)$, $(a,b)\ra (b,b)$. We are interested in the number of steps until all nodes become $b$s, i.e. the time needed for every node to interact at least once.

\begin{proposition} \label{pro:node-cover}
The expected time to convergence of a node cover is $\Theta(n\log n)$.
\end{proposition}
\begin{proof}
For the upper bound, simply observe that the running time of a one-to-all elimination, i.e. $\Theta(n\log n)$, is an upper bound on the running time of a node cover. The reason is that a node cover is a one-to-all elimination in which in some cases we may get two new bs by one effective transition (namely $(a,a)\ra (b,b)$) while in one-to-all elimination all effective transitions result in at most one new $b$.

For the lower bound, if $i$ is the number of $b$s then the probability of success is $p_i=1-[i(i-1)]/[n(n-1)]$. Observe now that a node cover process is slower than the artificial variation in which whenever rule $(a,b)\ra (b,b)$ applies we pick another $a$ and make it a $b$. This is because, given $i$ $b$s, this artificial process has the same probability of success as a node cover but additionally in every success the artificial process is guaranteed to produce two new $b$s while a node cover may in some cases produce only one new $b$. Define $k=\lceil n/2\rceil+1$. Then, taking into account what we already proved in the lower bound of one-to-all elimination (see Proposition \ref{pro:one-to-all}), we have
\begin{align*}
\E[X]&\geq n(n-1)\sum_{i=0}^{\lceil n/2\rceil} \frac{1}{n(n-1)-2i(2i-1)}
= \frac{n(n-1)}{4}\sum_{i=0}^{k-1} \frac{1}{\frac{n(n-1)}{4}-\frac{2i(2i-1)}{4}}\\
&= \frac{n(n-1)}{4}\sum_{i=0}^{k-1} \frac{1}{\frac{n}{2}(\frac{n}{2}-\frac{1}{2})-i(i-\frac{1}{2})}
> \frac{n(n-1)}{4}\sum_{i=0}^{k-1} \frac{1}{k(k-1)-i(i-1)}\\
&> \frac{n(n-1)}{8k}(H_{2k-2}-1)
> \frac{n-1}{8}(H_n-1)\\
&= \frac{n-1}{8}[\ln n+\Theta(1)].
\end{align*}
We conclude that $\E[X]=\Theta(n\log n)$. 
\qed
\end{proof}

\noindent\textbf{Edge cover.} All nodes are in state $a$ throughout the execution of the protocol. The only effective transition is $(a,a,0)\ra (a,a,1)$ (we now focus on edge-state updates), i.e. whenever an edge is found inactive it is activated (recall that initially all edges are inactive). We study the number of steps until all edges in $E_I$ become activated, which is equal to the time needed for all possible interactions to occur.

\begin{proposition} \label{pro:edge-cover}
The expected time to convergence of an edge cover is $\Theta(n^2\log n)$.
\end{proposition}
\begin{proof}
Given that $m=n(n-1)/2$ and given that $j$ successes (i.e. $j$ distinct interactions) have occurred the corresponding probability for the coupon collector argument is $p_j=(m-j)/m$ and the expected number of steps is $\E[X]=\sum_{i=0}^{m-1} m/(m-i)=m\sum_{i=0}^{m-1} 1/(m-i)=m\sum_{i=1}^{m} 1/i=m(\ln m + \Theta(1))=\Theta(n^2\log n)$. Another way to see this is to observe that it is a classical coupon collector problem with $m$ coupons each selected in every step with probability $1/m$, thus $\E[X]=m\ln m + O(m)=\Theta(n^2\log n)$.
\qed
\end{proof}

Table \ref{tab:basic-proc-full} summarizes the expected time to convergence of each of the above fundamental probabilistic processes.

\begin{table}[h]
\normalsize
\setlength{\tabcolsep}{15pt}
\begin{center}
\begin{tabular}{  l  c  }
  \hline
  Protocol & Expected Time \\ \hline
  \emph{One-way epidemic} & $\Theta(n\log n)$ \\ 
  \emph{One-to-one elimination} & $\Theta(n^2)$ \\ 
  \emph{One-to-all elimination} & $\Theta(n\log n)$ \\ 
  \emph{Meet everybody} & $\Theta(n^2\log n)$ \\
  \emph{Node Cover} & $\Theta(n\log n)$ \\ 
  \emph{Edge cover} & $\Theta(n^2\log n)$ \\ \hline
\end{tabular}
\end{center}
\caption{Our results for the expected time to convergence of several fundamental probabilistic processes.} \label{tab:basic-proc-full}
\end{table}

\section{Constructing a Global Line}
\label{sec:global-line}

In this section, we study probably the most fundamental network-construction problem, which is the problem of constructing a spanning line. Its importance lies in the fact that a spanning line provides an ordering on the processes which can then be exploited (as shown in Section \ref{sec:gencon}) to simulate a TM and thus to establish universality of our model. We give three different protocols for the spanning line problem each improving on the running time but using more states to this end.

We begin with a generic lower bound holding for all protocols that construct a spanning network.  

\begin{theorem} [Generic Lower Bound]
The expected time to convergence of any protocol that constructs a spanning network, i.e. one in which every node has at least one active edge incident to it, is $\Omega(n\log n)$. Moreover, this is the best lower bound for general spanning networks that we can hope for, as there is a protocol that constructs a spanning network in $\Theta(n\log n)$ expected time.
\end{theorem}
\begin{proof}
Consider the time at which the last edge is activated. Clearly, by that time, all nodes must have some active edge incident to them which implies that every node must have interacted at least once. Thus the running time is lower bounded by a node cover, which by Proposition \ref{pro:node-cover} takes an expected number of $\Theta(n\log n)$ steps.

Now consider the variation of node cover which in every transition that is effective w.r.t. node-states additionally activates the corresponding edge. In particular, the protocol consists of the rules $(a,a,0)\ra (b,b,1)$ and $(a,b,0)\ra (b,b,1)$. Clearly, when every node has interacted at least once, or equivalently when all $a$s have become $b$s, every node has an active edge incident to it, and thus the resulting stable network is spanning. The reason is that all nodes are $a$s in the beginning, every node at some point is converted to $b$, and every such conversion results in an activation of the corresponding edge. As a node-cover completes in $\Theta(n\log n)$ steps, the above protocol takes $\Theta(n\log n)$ steps to construct a spanning network.
\qed
\end{proof}

We now give an improved lower bound for the particular case of constructing a spanning line.

\begin{theorem} [Line Lower Bound]
The expected time to convergence of any protocol that constructs a spanning line is $\Omega(n^2)$.
\end{theorem}
\begin{proof}
Take any protocol $\ca$ that constructs a spanning line and any execution of $\ca$ on $n$ nodes. Consider the step $t$ at which $\ca$ performed the last modification of an edge. Observe that the construction after step $t$ must be a spanning line. We distinguish two cases. 

(i) The last modification was an activation. In this case, the construction just before step $t$ was either a line on $n-1$ nodes and an isolated node or two disjoint lines spanning all nodes. To see this, observe that these are the only constructions that can be turned into a line by a single additional activation. In the first case, the probability of obtaining an interaction between the isolated node and one of the endpoints of the line is $4/[n(n-1)]$ and in the second the probability of obtaining an interaction between an endpoint of one line and an endpoint of the other line is $8/[n(n-1)]$. In both cases, the expected number of steps until the last edge becomes activated is $\Omega(n^2)$. 

(ii) The last modification was a deactivation. This implies that the construction just before step $t$ was a spanning line with an additional active edge between two nodes, $u$ and $v$, that are not neighbors on the line. If one of these nodes, say $u$, is an internal node, then $u$ has degree 3 and we can only obtain a line by deactivating one of the edges incident to $u$. Clearly, the probability of getting one of these edges is $6/[n(n-1)]$ and it is even smaller if both nodes are internal. Thus, if at least one of $u$ and $v$ is internal, the expected number of steps is $\Omega(n^2)$. It remains to consider the case in which the construction just before step $t$ was a spanning ring, i.e. the case in which $u$ and $v$ are the endpoints of the spanning line. In this case, consider the step $t^\prime < t$ of the last modification of an edge that resulted in the ring. To this end notice that all nodes of a ring have degree 2. If $t^\prime$ was an activation then exactly two nodes had degree 1 and if $t^\prime$ was a deactivation then two nodes had degree 3. In both cases, there is a single interaction that results in a ring, the probability of success is $2/[n(n-1)]$ and the expectation is again $\Omega(n^2)$. 
\qed
\end{proof}

We proceed by presenting protocols for the spanning line problem.

\subsection{1st Protocol}
\label{subsec:simple-global-line}

We present now our simplest protocol for the spanning line problem.

\floatname{algorithm}{Protocol}
\renewcommand{\algorithmiccomment}[1]{// #1}
\begin{algorithm}[!h]
  \caption{\emph{Simple-Global-Line}}\label{prot:gline}
  \begin{algorithmic}
    \medskip
    \State $Q=\{q_0,q_1,q_2,l,w\}$
    \State $\delta$: 
    \begin{align*}
    (q_0,q_0,0)&\ra (q_1,l,1)\\
    (l,q_0,0)&\ra (q_2,l,1)\\
    (l,l,0)&\ra (q_2,w,1)\\
    (w,q_2,1)&\ra (q_2,w,1)\\
    (w,q_1,1)&\ra (q_2,l,1)
    \phantom{\hspace{10cm}}
    \end{align*}
    \State \Comment {All transitions that do not appear have no effect}
  \end{algorithmic}
\end{algorithm}

\noindent\emph{Protocol Simple-Global-Line:} $(q_0,q_0,0)\ra (q_1,l,1)$, $(l,q_0,0)\ra (q_2,l,1)$, $(l,l,0)\ra (q_2,w,1)$, $(w,q_2,1)\ra (q_2,w,1)$, $(w,q_1,1)\ra (q_2,l,1)$

\begin{theorem} \label{the:gline}
Protocol \emph{Simple-Global-Line} constructs a spanning line. It uses 5 states and its expected running time is $\Omega(n^4)$ and $O(n^5)$.
\end{theorem}
\begin{proof}
We begin by proving that, for any number of processes $n\geq 2$, the protocol correctly constructs a spanning line under any fair scheduler. Then we study the running time of the protocol under the uniform random scheduler.

\emph{Correctness.} In the initial configuration $C_0$, all nodes are in state $q_0$ and all edges are inactive, i.e in state 0. Every configuration $C$ that is reachable from $C_0$ consists of a collection of active lines and isolated nodes. Additionally, every active line has a unique leader which either occupies an endpoint and is in state $l$ or occupies an internal node, is in state $w$, and moves along the line. Whenever the leader lies on an endpoint of its line, its state is $l$ and whenever it lies on an internal node, its state is $w$. Lines can expand towards isolated nodes and two lines can connect their endpoints to get merged into a single line (with total length equal to the sum of the lengths of the merged lines plus one). Both of these operations only take place when the corresponding endpoint of every line that takes part in the operation is in state $l$. 

We have to prove two things: (i) there is a set $\cs$ of output-stable configurations whose active network is a spanning line, (ii) for every reachable configuration $C$ (i.e. $C_0\rsa C$) it holds that $C\rsa C_s$ for some $C_s\in \cs$. For (i), consider a spanning line, in which the non-leader endpoints are in state $q_1$, the non-leader internal nodes in $q_2$, and there is a unique leader either in state $l$ if it occupies an endpoint or in state $w$ if it occupies an internal node. For (ii), note that any reachable configuration $C$ is a collection of active lines with unique leaders and isolated nodes. We present a (finite) sequence of transitions that converts $C$ to a $C_s\in\cs$. If there are isolated nodes, take any line and if its leader is internal make it reach one of the endpoints by selecting the appropriate interactions. Then successively apply the rule $(l,q_0,0)\ra (q_2,l,1)$ to expand the line towards all isolated nodes. Thus we may now w.l.o.g. consider a collection of lines without isolated nodes. By successively applying the rule $(l,l,0)\ra (q_2,w,1)$ to pairs of lines while always moving the internal leaders that appear towards an endpoint it is not hard to see that the process results in an output-stable configuration from $\cs$, i.e. one whose active network is a spanning line. 

\emph{Running Time Upper Bound.} For the running time upper bound, we have an expected number of $O(n^2)$ steps until another progress is made (i.e. for another merging to occur given that at least two $l$-leaders exist) and $O(n^4)$ steps for the resulting random walk (walk of state $w$ until it reaches one endpoint of the line) to finish and to have again the system ready for progress. $O(n^4)$ follows because we have a random walk on a line with two absorbing barriers (see e.g. \cite{Fe68} pages 348-349) delayed on average by a factor of $O(n^2)$. As progress must be made $n-2$ times, we conclude that the expected running time of the protocol is bounded from above by $(n-2)[O(n^2)+O(n^4)]=O(n^5)$.

We next prove that we cannot hope to improve the upper bound on the expected running time by a better analysis by more than a factor of $n$. For this we first prove that the protocol w.h.p. constructs $\Theta(n)$ different lines of length 1 during its course. A set of $k$ disjoint lines implies that $k$ distinct merging processes have to be executed in order to merge them all in a common line and each single merging results in the execution of another random walk. We exploit all these to prove the desired $\Omega(n^4)$ lower bound.

Recall that initially all nodes are in $q_0$. Every interaction between two $q_0$-nodes constructs another line of length 1. Call the random interaction of step $i$ a \emph{success} if both participants are in $q_0$. Let $R$ be the r.v. of the number of nodes in state $q_0$; i.e. initially $R=n$. Note that, at every step, $R$ decreases by at most 2, which happens only in a success (it may also remain unchanged, or decrease by 1 if a leader expands towards a $q_0$). Let the r.v. $X_i$ be the number of successes up to step $i$ and $X$ be the total number of successes throughout the course of the protocol (e.g. until no further successes are possible or until stabilization). Our goal is to calculate the expectation of $X$ as this is equal to the number of distinct lines of length 1 that the protocol is expected to form throughout its execution (note that these lines do not necessarily have to coexist). Given $R$, the probability of success at the current step is $p_R=[R(R-1)]/[n(n-1)]\geq (R-1)^2/n^2$. As long as $R\geq (n/2)+1=z$ it holds that $p_R\geq (n^2/4)/n^2=1/4$. Moreover, as $R$ decreases by at most 2 in every step, there are at least $(n-z)/2=[(n/2)-1]/2=(n/4)-1/2$ steps until $R$ becomes less or equal to $z$. Thus, our process \emph{dominates} a Bernoulli process $Y$ with $(n/4)-1/2$ trials and probability of success $p^\prime=1/4$ in each trial. For this process we have $\E[Y]=[(n/4)-1/2](1/4)=(n/16)-1/8=\Theta(n)$.

We now exploit the following Chernoff bound (cf. \cite{MR95}, page 70) establishing that w.h.p. $Y$ does not deviate much below its mean $\mu=\E[Y]$:\\
\emph{
Chernoff Bound. Let $Y_1, Y_2,\ldots, Y_t$ be independent Poisson trials such that, for $1\leq i\leq t$, $\P[Y_i=1]=p_i$, where $0<p_i<1$. Then, for $Y=\sum_{i=1}^t Y_i$, $\mu=\E[Y]=\sum_{i=1}^t p_i$, and $0<\delta<1$,}
\begin{equation*}
\P[Y<(1-\delta)\mu]<\exp(-\mu\delta^2/2).
\end{equation*}

Additionally, it holds that $\exp(-\mu\delta^2/2)=\epsilon\Leftrightarrow \delta=\sqrt{\frac{2\ln 1/\epsilon}{\mu}}$. Thus $\exp(-\mu\delta^2/2)=n^{-c}$ implies $\delta^2=\frac{2c\ln n}{\mu}=$ $\frac{2c\ln n}{(1/8)(n/2-1)}=$ $\frac{16c\ln n}{n/2-1}\Rightarrow$ $\delta=\sqrt{\frac{16c\ln n}{n/2-1}}\Rightarrow$
\begin{align*}
(1-\delta)\mu&=\frac{1}{8}\left(1-\sqrt{\frac{16c\ln n}{n/2-1}}\right)\left(\frac{n}{2}-1\right) > \frac{1}{16}\left(n - 2\sqrt{cn\ln n} - 2\right) = \Theta(n).
\end{align*}

So, for all $c=O(1)$,
\begin{align*}
\P[Y<\frac{1}{16}\left(n - 2\sqrt{cn\ln n} - 2\right)]&<n^{-c}\Rightarrow\\
\P[Y\geq\frac{1}{16}\left(n - 2\sqrt{cn\ln n} - 2\right)]&>1-n^{-c}\Rightarrow\\
\P[Y=\Theta(n)]&>1-n^{-c}
\end{align*}
and as $X$ dominates $Y$, we have $\P[X=\Theta(n)]>1-n^{-c}$. In words, w.h.p. we expect $\Theta(n)$ lines of length 1 to be constructed by the protocol.

Now, given that $X=\Theta(n)$, we distinguish two cases: (i) At some point during the course of the protocol two lines both of length $\Theta(n)$ get merged. In this case, the corresponding random walk takes on average $\Theta(n^2)$ transitions involving the leader and on average the leader is selected every $\Theta(n^2)$ steps to interact with one of the 2 active edges incident to it. That is, the expected number of steps for the completion of such a random walk is $\Theta(n^4)$ and the expected running time of the protocol is $\Omega(n^4)$ (ii) In every merging process, at least one of the two participating lines has length at most $d_{max}=o(n)$. We have already shown that the protocol w.h.p. forms $k=\Theta(n)$ distinct lines of length 1. Consider now the interval $I=\{k/2-d_{max}+1,\ldots ,k/2\}$. As $h=\Theta(n)>d_{max}$ for all $h\in I$, only a single line can ever have length $h\in I$ and one, call it $l_1$, will necessarily fall in this interval due to the fact that the length of $l_1$ will increase by at most $d_{max}$ in every merging until it becomes $n$. Consider now the time at which $l_1$ has length $h\in I$. As the total length due to lines of length 1 (ever to appear) is $k$ and the length of $l_1$ is $h$ there is still a remaining length of at least $k-h\geq k-k/2=k/2=\Theta(n)$ to be merged to $l_1$. As the maximum length of any line different than $l_1$ is $d_{max}$, $l_1$ will get merged to the $k-h$ remaining length via at least $j=\Theta(n)/d_{max}$ distinct mergings with lines of length at most $d_{max}$. These mergings, and thus also the resulting random walks, cannot occur in parallel as all of them share $l_1$ as a common participant (and a line can only participate in one merging at a time). Let $d_i$ denote the length of the $i$-th line merged to $l_1$, for $1\leq i\leq j$. If $l_1$ has length $d(l_1)$ just before the $i$-th merging, then the expected duration of the resulting random walk is $n^2\cdot d(l_1)\cdot d_i$ and the new $l_1$ resulting from merging will have length $d(l_1)+d_i$. Let $Y$ denote the duration of all random walks, and $Y_i$, $1\leq i\leq j$, the duration of the $i$-th random walk. In total, the expected duration of all random walks resulting from the $j$ mergings of $l_1$ is
\begin{align*}
\E[Y]&=\E[\sum_{i=1}^j Y_i]=\sum_{i=1}^j \E[Y_i]\\
&=\sum_{i=1}^j n^2(h+d_1+\ldots+d_{i-1})d_i\\
&\geq n^2\sum_{i=1}^j hd_i= n^2h\sum_{i=1}^j d_i\\
&=n^2\Theta(n)\Theta(n)\\
&=\Theta(n^4).
\end{align*}
The fifth equality follows from the fact that $\sum_{i=1}^j d_i=k-h=\Theta(n)$. We conclude that the expected running time of the protocol is also in this case $\Omega(n^4)$.

Now if we define the r.v. $W$ to be the total running time of the protocol (until convergence), by the law of total probability and for every constant $c\geq 1$, we have that:
\begin{align*}
\E[W]&=\E[W\mid X=\Theta(n)]\cdot\P[X=\Theta(n)]+\E[W\mid X=o(n)]\cdot\P[X=o(n)]\\
&\geq \E[W\mid X=\Theta(n)]\cdot\P[X=\Theta(n)] >\Theta(n^4)(1-n^{-c})=\Theta(n^4-n^{4-c})\\
&=\Theta(n^4).
\end{align*}
Thus, the expected running time of the protocol is $\Omega(n^4)$.
\qed
\end{proof}

\subsection{2nd Protocol}
\label{subsec:gline2}

The random walk approach followed in Protocol \ref{prot:gline} takes time, thus a straightforward attempt for improvement is to replace the random walk merging process with some more ``deterministic'' merging. In Protocol \ref{prot:gline2}, the random walk rules 3-5 of Protocol \ref{prot:gline} have been replaced by a more ``deterministic'' procedure.  

\floatname{algorithm}{Protocol}
\renewcommand{\algorithmiccomment}[1]{// #1}
\begin{algorithm}[!h]
  \caption{\emph{Intermediate-Global-Line}}\label{prot:gline2}
  \begin{algorithmic}
    \medskip
    \State $Q=\{q_0,q_1,q_2,l,\bar{w},w_1,w_2,w_3\}$
    \State $\delta$: 
    \begin{align*}
    (q_0,q_0,0)&\ra (q_1,l,1)\\
    (l,q_0,0)&\ra (q_2,l,1)\\
    (l,l,0)&\ra (\bar{w},w_1,1)\\
    (w_1,q_2,1)&\ra (w_1,w_1,1)\\
    (w_1,q_1,1)&\ra (w_2,q_1,1)\\
    (w_2,w_1,1)&\ra (w_2,w_2,1)\\
    (w_2,\bar{w},1)&\ra (w_3,q_2,1)\\
    (w_3,w_2,1)&\ra (q_2,w_3,1)\\
    (w_3,q_1,1)&\ra (q_2,l,1)
    \phantom{\hspace{10cm}}
    \end{align*}
  \end{algorithmic}
\end{algorithm}

\begin{theorem} \label{the:gline2}
Protocol \emph{Intermediate-Global-Line} constructs a spanning line. It uses 8 states and its expected running time under the uniform random scheduler is $\Omega(n^3)$ and $O(n^4)$.
\end{theorem}
\begin{proof}
The proof idea is precisely the same as that of Theorem \ref{the:gline}. The only difference is that now merging two lines of lengths $d_1$ and $d_2$ takes time $\max\{d_1,d_2\}\simeq d_1+d_2$ (asymptotically) instead of the $d_1\cdot d_2$ of the random walk. Thus, for the upper bound we need $O(n)$ mergings each taking an average of $O(n^3)$ to complete in the worst case. $O(n^3)$ holds because $O(n)$ steps are performed by the merging process in the worst-case and the process must wait an average of $O(n^2)$ until its leader is selected to interact over one of its active edges. Thus the dominating factor is now $O(n^4)$.

For the lower bound we again have w.h.p. $\Theta(n)$ lines of length 1 and for cases (i), (ii) as above we have: (i) merging two lines both of length $\Theta(n)$ takes time $\Theta(n^3)$. (ii)
\begin{align*}
\E[Y]&=\E[\sum_{i=1}^j Y_i]=\sum_{i=1}^j \E[Y_i]=\sum_{i=1}^j n^2(h+d_1+\ldots+d_{i-1})\\
&\geq n^2 hj=n^2\Theta(n)\Theta(n)/d_{max}=\Theta(n^4)/o(n).
\end{align*}
\qed
\end{proof}

\subsection{3rd Protocol}
\label{subsec:third-line}

We now give our fastest protocol (Protocol \ref{prot:gline3}) for the global line construction. The main difference between this and the previous protocols is that we now totally avoid mergings as they seem to consume much time. As shown above, even a merging in which a linear number of steps must be performed needs $\Theta(n^3)$ time as every step takes an average of $\Theta(n^2)$ time. Then a linear number of mergings naturally require an average of $\Theta(n^4)$ time, which is quite big. The intuition behind the following improvement is that when two disjoint lines interact, instead of merging, the corresponding leaders play a game in which only one survives. The winner grows by one towards the other line and the loser sleeps. A sleeping line cannot increase any more and only loses nodes by lines that are still awake. A single leader is guaranteed to always win and this occurs quite fast. Then the leader makes progress (by one) in most interactions and every such progress is in turn quite fast.

\floatname{algorithm}{Protocol}
\renewcommand{\algorithmiccomment}[1]{// #1}
\begin{algorithm}[!h]
  \caption{\emph{Fast-Global-Line}}\label{prot:gline3}
  \begin{algorithmic}
    \medskip
    \State $Q=\{q_0,q_1,q_2,q_2^\prime,l,l^\prime,l^\dprime,f_0,f_1\}$
    \State $\delta$: 
    \begin{align*}
    (q_0,q_0,0)&\ra (q_1,l,1)\\
    (l,q_0,0)&\ra (q_2,l,1)\\
    (l,l,0)&\ra (q_2^\prime,l^\prime,1)\\
    (l^\prime,q_2,1)&\ra (l^\dprime,f_1,0)\\
    (l^\prime,q_1,1)&\ra (l^\dprime,f_0,0)\\
    (l^\dprime,q_2^\prime,1)&\ra (l,q_2,1)\\
    (l,f_0,0)&\ra (q_2,l,1)\\
    (l,f_1,0)&\ra (q_2^\prime,l^\prime,1)
    \phantom{\hspace{10cm}}
    \end{align*}
  \end{algorithmic}
\end{algorithm}

\begin{theorem} \label{the:gline3}
Protocol \emph{Fast-Global-Line} constructs a spanning line. It uses 9 states and its expected running time under the uniform random scheduler is $O(n^3)$.
\end{theorem}
\begin{proof}
Observe first that in $O(n^2)$ steps all $q_0$s become something else. To see this let $X$ be the r.v. of the total number of steps until all $q_0$s disappear and let $X_i$ be the r.v. of the number of steps between the $i$th and the $(i+1)$st interaction between two nodes in state $q_0$ (assume no other interactions can change the state of a $q_0$). Let $p_i=[(n-2i)(n-2i-1)]/[n(n-1)]$ be the probability that such an interaction occurs. Then $\E[X_i]=1/p_i=\Theta(n^2/(n-i)^2)$ and $\E[X]\simeq n^2\sum_{i=1}^{n/2} 1/(n-i)^2=\Theta(n^2)$. The last equation follows from the fact that $\sum_{i=1}^{n/2} 1/(n-i)^2\leq\sum_{i=1}^{n^2} 1/i-\sum_{i=1}^{(n/2)^2} 1/i\simeq 2\ln n+\Theta(1) -2\ln n +2\ln 2 -\Theta(1) = O(1)$, i.e. it is bounded. Finally, observe that $q_0$s that become leaders can also turn other $q_0$s to something else thus the actual expectation is in fact $O(n^2)$ (i.e. what we have ignored can only help the process end faster).

Now notice that after this $O(n^2)$ time we have a set of at most $O(n)$ leaders and no new leader can ever appear. Moreover, in every interaction between two leaders only one survives and the other becomes a follower. Clearly, a single leader must win all the pairwise games in which it will participate. Consider that leader and observe that it takes it an average of $n^2$ steps to participate to another game in the worst case and another $n^2$ steps to win it. Clearly, in $O(n^2)$ steps on average there is a unique leader and every other node is either isolated in state $f_0$ or part of a line that has a unique follower $f_1$. Every interaction of a leader with a follower increases the length of the leader's line by 1 in $O(n^2)$ steps. Thus an increment occurs every $O(n^2)$ steps as the leader needs $O(n^2)$ steps to meet a follower and then $O(n^2)$ steps to increase by 1 towards that follower. As the leader needs to make at most $O(n)$ increments to make its own line global, we conclude that the expected time for this to occur is $O(n)\cdot O(n^2)=O(n^3)$.
\qed
\end{proof}

\section{Other Basic Constructors}
\label{sec:basic-con}

\noindent\textbf{Cycle Cover}

\floatname{algorithm}{Protocol}
\renewcommand{\algorithmiccomment}[1]{// #1}
\begin{algorithm}[!h]
  \caption{\emph{Cycle-Cover}}\label{prot:cycle-cover}
  \begin{algorithmic}
    \medskip
    \State $Q=\{q_0,q_1,q_2\}$
    \State $\delta$: 
    \begin{align*}
    (q_0,q_0,0)&\ra (q_1,q_1,1)\\
    (q_1,q_0,0)&\ra (q_2,q_1,1)\\
    (q_1,q_1,0)&\ra (q_2,q_2,1)
    \phantom{\hspace{10cm}}
    \end{align*}
  \end{algorithmic}
\end{algorithm}

\begin{theorem} \label{the:cycle-cover}
Protocol \emph{Cycle-Cover} constructs a cycle cover with waste 2 (i.e. a cycle cover on a subset of $V_I$ of $n-2$ nodes). It uses 3 states, its expected running time under the uniform random scheduler is $\Theta(n^2)$, and it is optimal w.r.t. time.
\end{theorem}
\begin{proof}
Note that the protocol stabilizes when all nodes have become $q_2$. In $O(n^2)$ time all $q_0$s have become $q_1$ and in another $O(n^2)$ steps (by dominating a one-to-one elimination) all $q_1$s have become $q_2$s. For the lower bound consider the last edge modification that ever occurs. Due to the symmetry of cycle cover, both if it was an activation or a deactivation only a single edge satisfies the fact that after its activation or deactivation we get a cycle cover, which requires $\Theta(n^2)$ rounds (this is a lower bound for any protocol that constructs a cycle cover). That the waste is 2 follows from the fact that some executions may construct a cycle cover on $n-2$ nodes and leave the remaining two nodes connected and in state $q_1$ forever.
\qed
\end{proof}

\noindent\textbf{Global Star}

\begin{theorem} [Star Lower Bound]
Any protocol that constructs a spanning star has at least 2 states and its expected time to convergence is $\Omega(n^2\log n)$.
\end{theorem}
\begin{proof}
Clearly, with a single state we cannot make the necessary distinction of a center and a peripheral node. More formally, if there is a single state $q_0$ then $(q_0,q_0,0)$ must necessarily activate the edge (otherwise no edges will be ever activated) which implies that eventually all edges will become activated, i.e. instead of a star we will end up with a global clique. So every protocol that constructs a global star must have at least 2 states.

For the lower bound on the expected running time we argue as follows. Take any execution of a protocol that constructs a global star. Consider the node $u$ that will become the center in that execution. When the execution stabilizes, $u$ must be connected to every other node by an active edge. This implies that $u$ must have interacted to every other node. Clearly, the time it takes for the eventually unique center, $u$ in this case, to meet every other node is a lower bound on the total running time. This is a meet everybody that, as proved in Proposition \ref{pro:meet-ever}, takes $\Theta(n^2\log n)$ time.   
\qed
\end{proof}

\floatname{algorithm}{Protocol}
\renewcommand{\algorithmiccomment}[1]{// #1}
\begin{algorithm}[!h]
  \caption{\emph{Global-Star}}\label{prot:gstar}
  \begin{algorithmic}
    \medskip
    \State $Q=\{c,p\}$, $q_0=c$
    \State $\delta$: 
    \begin{align*}
    (c,c,0)&\ra (c,p,1)\\
    (p,p,1)&\ra (p,p,0)\\
    (c,p,0)&\ra (c,p,1)
    \phantom{\hspace{10cm}}
    \end{align*}
  \end{algorithmic}
\end{algorithm}

\begin{theorem}
Protocol \emph{Global-Star} constructs a spanning star. It uses 2 states and its expected running time under the uniform random scheduler is $O(n^2\log n)$, that is it is optimal both w.r.t. size and time.
\end{theorem}
\begin{proof}
\emph{Correctness.} Each node may play one of the following two roles during an execution of the protocol: a \emph{center} (state $c$) or a \emph{peripheral} (state $p$). The unique output-stable configuration $C_f$ whose active network is a spanning star, has one center and $n-1$ peripheral nodes, and a $uv$ edge is active iff one of $u,v$ is the center. Initially all nodes are centers. When two centers interact one of them remains a center and the other becomes a peripheral. No other interactions eliminate a center, which implies that not all centers can be eliminated, and once a center becomes a peripheral it can never become a center again. Due to fairness, eventually all pairs of centers will interact and, as no new centers appear, eventually a single center will remain. Thus from some point on there is a single center and $n-1$ peripheral nodes. The idea from now on is that $c$-$p$ attract while $p$-$p$ repel. In particular, rule $(c,p,0)\ra (c,p,1)$ guarantees that any inactive edges joining the center to the peripherals will become activated and rule $(p,p,1)\ra (p,p,0)$ guarantees that any active edges joining two peripherals will become deactivated. At the same time active edges between the center and the peripherals remain active and inactive edges between two peripherals remain inactive. This clearly leads to the construction of a spanning star.

\emph{Running Time.} Forget for a while the edge updates and consider the rule $(c,c)\ra (c,p)$, which is the only effective interaction of the protocol w.r.t. the states of the nodes. We are interested in the time needed for a single $c$ to remain. This is clearly an original application of one-to-one elimination and as proved in Proposition \ref{pro:one-to-one} it takes $\Theta(n^2)$ time.

Notice now that once the states of the nodes have stabilized, the constructed network will for sure stabilize to a global star after all $p$-nodes have interacted with each other in order to deactivate any active edges between them and after the $c$ has interacted with all $p$s in order to activate any inactive edges, i.e. after all pairs of interactions have occurred. This is an edge cover that, as proved in Proposition \ref{pro:edge-cover}, takes $\Theta(n^2\log n)$ time. Thus the total expected running time is at most $\Theta(n^2)+\Theta(n^2\log n)=\Theta(n^2\log n)$.
\qed
\end{proof}

\noindent\textbf{Global Ring}

\begin{theorem} [Ring Lower Bound]
The expected time to convergence of any protocol that constructs a spanning ring is $\Omega(n^2)$.
\end{theorem}
\begin{proof}
Take any protocol $\ca$ that constructs a spanning ring and any execution of $\ca$ on $n$ nodes. Consider the step $t$ at which $\ca$ performed the last modification of an edge. Observe that the construction after step $t$ must be a spanning ring. We distinguish two cases. 

(i) The last modification was an activation. It follows that the previous active network should be a spanning line $u_1,u_2,\ldots,u_n$. But the only activation that can convert this spanning line into a spanning ring is $u_1u_n$ which occurs with probability $2/[n(n-1)]$, i.e. in an expected number of $\Theta(n^2)$ steps.  

(ii) The last modification was a deactivation. It follows that the previous active network should be a spanning ring $u_1,u_2,\ldots,u_n,u_1$ with an additional active edge $u_iu_j$ for $1\leq i<j\leq n$ and $j\neq i+1$ (i.e. a chord). Clearly, the only interaction that can convert such an active network into a spanning ring is $u_iu_j$ which takes an expected number of $\Theta(n^2)$ steps to occur.
\qed
\end{proof}

\floatname{algorithm}{Protocol}
\renewcommand{\algorithmiccomment}[1]{// #1}
\begin{algorithm}[!h]
  \caption{\emph{Global-Ring}}\label{prot:gring}
  \begin{algorithmic}
    \medskip
    \State $Q=\{q_0,q_1,q_2,l,w,l^\prime,\l^\dprime,q_1^\prime,q_1^\dprime\}$
    \State $\delta$: 
    \begin{align*}
    (q_0,q_0,0)&\ra (q_1,l,1)\\
    (l,q_0,0)&\ra (q_2,l,1)\\
    (l,l,0)&\ra (q_2,w,1)\\
    (w,q_2,1)&\ra (q_2,w,1)\\
    (w,q_1,1)&\ra (q_2,l,1)\\
    (l,q_1,0)&\ra (l^\prime,q_1^\prime,1)\\
    (x^\prime,y,0)&\ra (x^\dprime,y,0)\text{, for } x\in\{l,q_1\} \text{ and } y\in\{l,w,q_1,q_0\}\\
    (x^\prime,y^\prime,0)&\ra (x^\dprime,y^\dprime,0)\text{, for } x\in\{l,q_1\} \text{ and } y\in\{l,q_1\}\\
    (l^\dprime,q_1^\prime,1)&\ra (l,q_1,0)\\
    (l^\prime,q_1^\dprime,1)&\ra (l,q_1,0)\\
    (l^\dprime,q_1^\dprime,1)&\ra (l,q_1,0)
    \phantom{\hspace{10cm}}
    \end{align*}
    \State \Comment {The first 5 rules are the same as in the \emph{Simple-Global-Line} protocol (Protocol \ref{prot:gline})} 
  \end{algorithmic}
\end{algorithm}

\begin{theorem}
Protocol \emph{Global-Ring} constructs a spanning ring.
\end{theorem}
\begin{proof}
The protocol is essentially the same as the \emph{Simple-Global-Line} protocol (Protocol \ref{prot:gline}) but additionally we allow the endpoints of a line to become connected. This occurs whenever one endpoint is in state $l$ and the other is in state $q_1$ and the two endpoints interact. In this case, rule $(l,q_1,0)\ra (l^\prime,q_1^\prime,1)$ applies and the two endpoints become blocked. If any of the two endpoints detects the existence of another component, then, in the next interaction between them, the two endpoints backtrack, by which we mean that they deactivate the connection between them and both become unblocked again by returning to their original states. The existence of another component can be eventually detected due to the fact that every component is either an isolated node in state $q_0$ or has at least one leader. Now take an arbitrary reachable configuration $C$ with at least 2 components. We may w.l.o.g. assume that $C$ has no blocked nodes, as if it has there is a sequence of interactions that unblocks them all. Thus, as in the \emph{Simple-Global-Line} protocol we have a collection of lines and isolated nodes. This may very well lead to the formation of a spanning line with a single leader. It is now clear that at some point the leader will occupy one endpoint of the line, will interact with the other endpoint, the spanning line will close to form a spanning ring and the previous endpoints will become blocked. As there is a single component in the network, these two nodes will remain blocked forever and therefore the constructed active ring is stable. 
\qed
\end{proof}

\noindent\textbf{Global Ring: A Generic Approach}

\floatname{algorithm}{Protocol}
\renewcommand{\algorithmiccomment}[1]{// #1}
\begin{algorithm}[!h]
  \caption{\emph{2RC}}\label{prot:2rc}
  \begin{algorithmic}
    \medskip
    \State $Q=\{q_0,q_1,q_2,l_1,l_2,l_3\}$
    \State $\delta$: 
    \begin{align*}
    (q_0,q_0,0)&\ra (q_1,l_1,1)\\
    (q_1,q_0,0)&\ra (q_2,q_1,1)\\
    (q_1,q_1,0)&\ra (q_2,q_2,1)\\
    (l_1,l_1,0)&\ra (l_2,q_2,1)\\
    (l_1,q_i,0)&\ra (q_2,l_{i+1},1)\text{, for } i\in\{0,1\}\\
    &\hspace{-60pt}\text{// swapping: leaders keep moving inside components}\\
    (l_i,q_j,1)&\ra (q_i,l_j,1)\text{, for } i,j\in\{1,2\}\\
    &\hspace{-60pt}\text{// leader elimination: eventually a single leader will remain in every component}\\
    (l_i,l_j,1)&\ra (q_i,l_j,1)\text{, for } i,j\in\{1,2\}\\
    &\hspace{-60pt}\text{// opening cycles in the presence of other components}\\
    (l_2,q_0,0)&\ra (l_3,q_1,1)\\
    (l_2,l_1,0)&\ra (l_3,q_2,1)\\
    (l_2,l_2,0)&\ra (l_3,l_3,1)\\
    (l_3,q_1,1)&\ra (l_2,q_0,0)\\
    (l_3,q_2,1)&\ra (l_2,l_1,0)\\
    (l_3,l_1,1)&\ra (l_2,q_0,0)\\
    (l_3,l_2,1)&\ra (l_2,l_1,0)\\
    (l_3,l_3,1)&\ra (l_2,l_2,0)
    \phantom{\hspace{8cm}}
    \end{align*} 
  \end{algorithmic}
\end{algorithm}

\begin{theorem} \label{the:2rc}
Protocol \emph{2RC} constructs a connected spanning 2-regular network (i.e. a spanning ring).
\end{theorem}
\begin{proof}
The set $\cs$ of output-stable configurations whose active network is a spanning ring consists of those configurations that have one node in state $l_2$ and all other nodes in state $q_2$. The index of a state indicates the number of active neighbors of a node. A first goal is for all nodes to have degree 2 which implies a cycle cover, i.e. a partitioning of the nodes into disjoint cycles. The protocol achieves this by allowing every node with degree smaller than 2 to increase its degree. The final goal is to end up with a unique spanning ring. To achieve this, the protocol allows nodes with degree 2 to drop an existing neighbor and pick a new one provided that there are at least 2 components in the network. Clearly, this implies that any closed cycle coexisting with other components, which are cycles, lines, or isolated nodes, may open to form a line. As any collection of lines and isolated nodes can always be merged to a global line and any global line can close to form a global ring, the theorem follows. 
\qed
\end{proof}

\noindent\textbf{Generalizing to $k$-Regular Connected}

\floatname{algorithm}{Protocol}
\renewcommand{\algorithmiccomment}[1]{// #1}
\begin{algorithm}[!h]
  \caption{\emph{$k$RC}}\label{prot:krc}
  \begin{algorithmic}
    \medskip
    \State $Q=\{q_0,q_1,\ldots,q_k,l_1,l_2,\ldots,l_{k+1}\}$, i.e. $|Q|=2(k+1)$
    \State $\delta$: 
    \begin{align*}
    (q_0,q_0,0)&\ra (q_1,l_1,1)\\
    (q_i,q_j,0)&\ra (q_{i+1},q_{j+1},1)\text{, for } 1\leq i<k \text{ and } j<k\\
    (l_i,l_j,0)&\ra (l_{i+1},q_{j+1},1)\text{, for } 1\leq i,j<k\\
    (l_i,q_j,0)&\ra (q_{i+1},l_{j+1},1)\text{, for } 1\leq i<k \text{ and } j<k\\
    &\hspace{-60pt}\text{// swapping: leaders keep moving inside components}\\
    (l_i,q_j,1)&\ra (q_i,l_j,1)\text{, for } 1\leq i,j\leq k\\
    &\hspace{-60pt}\text{// leader elimination: eventually a single leader will remain in every component}\\
    (l_i,l_j,1)&\ra (q_i,l_j,1)\text{, for } 1\leq i,j\leq k\\
    &\hspace{-60pt}\text{// opening $k$-regular components in the presence of other components}\\
    (l_k,q_0,0)&\ra (l_{k+1},q_1,1)\\
    (l_k,l_i,0)&\ra (l_{k+1},q_{i+1},1)\text{, for } 1\leq i<k\\
    (l_k,l_k,0)&\ra (l_{k+1},l_{k+1},1)\\
    (l_{k+1},q_1,1)&\ra (l_k,q_0,0)\\
    (l_{k+1},q_i,1)&\ra (l_k,l_{i-1},0)\text{, for } 2\leq i\leq k\\
    (l_{k+1},l_i,1)&\ra (l_k,l_{i-1},0)\text{, for } 1\leq i\leq k\\
    (l_{k+1},l_{k+1},1)&\ra (l_k,l_k,0)
    \phantom{\hspace{8cm}}
    \end{align*} 
  \end{algorithmic}
\end{algorithm}

Using almost the same ideas as in the proof of Theorem \ref{the:2rc}, one can prove the following. 

\begin{theorem}
For every fixed integer $k\geq 2$ and population of size $n\geq k+1$, Protocol \emph{$k$RC} (see Protocol \ref{prot:krc}) constructs a connected spanning network in which at least $n-k+1$ nodes have degree $k$ and each of the remaining $l\leq k-1$ nodes has degree at least $l-1$ and at most $k-1$.
\end{theorem}

It is interesting to point out that the number of states can be substantially reduced in some cases by relying on the computability of the target-degree $k$. For an example, we show that we can make a node $u$ obtain $2^d$ neighbors by using only $2(d+2)$ states, for all fixed integers $d$. Node $u$ is initially in state $q_0$ and all other nodes are in state $a_0$. The protocol is $(q_0,a_0,0)\ra (q_0^\prime,a_1,1)$, $(q_0^\prime,a_0,0)\ra (q,a_1,1)$, $(q,a_i,1)\ra (q_{i+1},a_{i+1},1)$, $(q_j,a_0,0)\ra (q,a_j,1)$ for all $1\leq i\leq d-1$, $2\leq j\leq d$. Note that $u$ initially collects 2 neighbors (by activating edges) which go to state $a_1$. Then for every $a_1$ neighbor that it encounters it makes it an $a_2$ and collects another neighbor which goes to state $a_2$. Eventually both $a_1$ neighbors will become $a_2$ and there will be another $2$ neighbors in state $a_2$, so in total 4 $a_2$ neighbors. This process is repeated $d$ times (the 4 $a_2$s will become 8 $a_3s$, and so on), each time doubling the number of neighbors, thus eventually $u$ will have obtained $2^d$ neighbors. The protocol uses only $2(d+1)$ states for the indices of the $q_i$s and the $a_i$s and another 2 states, namely $q$ and $q_0^\prime$. Clearly, it follows that the target-degree of the nodes is not a lower bound on the size of the protocol.\\

\noindent\textbf{Many Small Components}\\

We show here how to partition the population into small active cliques. This construction is of special value as such a partitioning may serve as a means of maintaining non-interfering clusters. In particular, given such a partitioning, we can easily have a node $u$ perform effective interactions only with nodes belonging to the same component as $u$. This can be easily determined by the state of the connection between the interacting nodes. 

\floatname{algorithm}{Protocol}
\renewcommand{\algorithmiccomment}[1]{// #1}
\begin{algorithm}[!h]
  \caption{\emph{$c$-Cliques}}\label{prot:cliques}
  \begin{algorithmic}
    \medskip
    \State $Q=\{l_0,l_1,\ldots,l_{c-2},f_1,\ldots,f_{c-2},f,\bar{l}_0,\ldots,\bar{l}_{c-2},l,1,2,\ldots,c-1,l_1^\prime,\ldots,l_{c-1}^\prime,r\}$
    \State $\delta$: 
    \begin{align*}
    (l_0,l_0,0)&\ra (l_1,f,1)\\
    (l_i,l_0,0)&\ra (l_{i+1},f,1), \text{ if } i<c-2\\
	       &\ra (\bar{l}_1,1,1), \text{ if } i=c-2\\
    (l_i,l_j,0)&\ra (l_{i+1},f_j,1), \text{ if } j\leq i<c-2\\
	       &\ra (\bar{l}_0,f_j,1), \text{ if } i=c-2\\
    (f_i,f,1)&\ra (f_{i-1},l_0,0), \text{ if } i>1\\
	     &\ra (f,l_0,0), \text{ if } i=1\\
    (\bar{l}_i,f,1)&\ra (\bar{l}_{i+1},1,1), \text{ if } i<c-2\\
	           &\ra (l,1,1), \text{ if } i=c-2\\
    (i,j,0)&\ra (i+1,j+1,1), \text{ if } i<c-1 \text{ and } j<c-1\\
    (l,i,1)&\ra (r,l_i^\prime,1)\\
    (l_i^\prime,l_j^\prime,1)&\ra (l_{i-1}^\prime,l_{j-1}^\prime,0)\\
    (l_i^\prime,r,1)&\ra (i,l,1)
    \phantom{\hspace{10cm}}
    \end{align*}
  \end{algorithmic}
\end{algorithm}

\begin{theorem}
For every fixed positive integer $c$, Protocol \emph{$c$-Cliques} (see Protocol \ref{prot:cliques}) constructs $\lfloor n/c\rfloor$ cliques of order $c$ each.
\end{theorem}
\begin{proof}
The protocol first constructs $\lfloor n/c\rfloor$ components of order $c$, each having a unique leader (state `$l$') directly connected to $c-1$ followers (state $f$). Each follower then tries to become connected to the other $c-1$ followers of the component. To do this, it becomes connected one after the other to $c-1$ other followers. As it cannot distinguish the followers of its component from the followers of other components several of these connections may be wrong. It suffices to prove that the protocol recognizes wrong connections and deactivates them. Then, as followers always try to make their degree $c-1$ when it is still less than $c-1$ and as wrong connections between different components are always corrected, it follows (by fairness) that eventually each component will become a clique (having only correct connections). At that time, no new connections may be created and no existing connection can be deactivated (as they are all correct), and the correctness of the protocol follows. To recognize erroneous connections, the leader of a component constantly visits the followers of its component and checks any active connections that it may encounter during its stay (the duration of its stay is nondeterministic as it depends on the chosen interactions). If it ever encounters another leader over an active connection, then this is clearly a connection between different components and the leaders deactivate that connection and decrease the counters of the corresponding followers. Clearly, by fairness, every wrong connection will eventually be selected for interaction while having a leader in each of its endpoints. Finally, note that correct connections (between nodes of the same component) are never deactivated as at any time at most one of their endpoints may be occupied by a leader.
\qed
\end{proof}

\noindent\textbf{Replication}\\

We now study the related problem of replicating a given input graph $G_1=(V_1,E_1)$. Let $V_2=V_I\bs V_1$ be the set of the remaining nodes. The protocols use the nodes of $V_2$ to construct a replica $G_2$ of $G_1$, thus it must hold that $|V_2|\geq |V_1|$. In most cases, it suffices that $V_2=V_1$. Throughout the section we assume that nodes in $V_1$ are in different initial states than nodes in $V_2$. We usually use $q_0$ and $r_0$ as the initial states of nodes in $V_1$ and $V_2$, respectively. Additionally, $E_1$ is defined by the active edges between nodes in $V_1$. We assume that $G_1$ is connected.

We present a very simple protocol (Protocol \ref{prot:lreplic}) that exploits the election of a unique leader to copy $G_1$ on $|V_2|=|V_1|$ free nodes.

\floatname{algorithm}{Protocol}
\renewcommand{\algorithmiccomment}[1]{// #1}
\begin{algorithm}[!h]
  \caption{\emph{Leader-Replication}}\label{prot:lreplic}
  \begin{algorithmic}
    \medskip
    \State $Q=\{q_0,r_0,l,l_a,l_d,f,f_a,f_d,r,r_a,r_d,r^\prime\}$
    \State $\delta$: 
    \begin{align*}
    (q_0,r_0,0) &\ra (l,r,1)\\
    (l,l,x) &\ra (l,f,x)\\
    (l,f,0) &\stackrel{1/2}\ra(l_d,f_d,0)\\
	    &\stackrel{1/2}\ra(f,l,0)\\
    (l,f,1) &\stackrel{1/2}\ra(l_a,f_a,0)\\
	    &\stackrel{1/2}\ra(f,l,1)\\
    (x_i,r,1) &\ra (x_i,r_i,1), \text{ for } x\in\{l,f\} \text{ and } i\in\{a,d\}\\
    (r_a,r_a,0) &\ra (r^\prime,r^\prime,1)\\
    (r_d,r_d,1) &\ra (r^\prime,r^\prime,0)\\
    (r^\prime,x_i,1) &\ra (r,x,1)\\
    (l_i,l_j,x) &\ra (l_i,f_j,x), \text{ for } i,j\in\{a,d\}
    \phantom{\hspace{8cm}}
    \end{align*}
  \end{algorithmic}
\end{algorithm}

\begin{theorem}
Protocol \emph{Leader-Replication} constructs a copy of any connected input graph $G_1=(V_1,E_1)$ with no waste. It uses 12 states and its expected running time under the uniform random scheduler is $\Theta(n^4\log n)$.
\end{theorem}
\begin{proof}
Initially, all nodes of $V_1$ are in $q_0$ and all nodes of $V_2$ are in $r_0$. The protocol first creates a matching between $V_1$ and $V_2$. Then it starts pairwise eliminations between leaders, that is when two leaders (nodes in state `$l$') interact one of them survives (i.e. remains $l$) and the other becomes a follower (state `$f$'). Eventually the protocol ends up with a unique leader and $|V_1|-1$ followers. Moreover, when a leader and a follower meet they swap their states with probability $1/2$. With the remaining $1/2$ probability they become either $l_a,f_a$ or $l_d,f_d$ depending on whether the edge joining them was active or inactive, respectively. In both cases they mark their matched nodes from $V_2$ to either activate or deactivate the edge between them in $V_2$ accordingly. Once there is a unique leader, the leader moves nondeterministically over the nodes of $V_1$ and again nondeterministically applies this copying process on the edges of $E_1$. Thus it will eventually apply this copying process to all edges of $E_1$ and as there are no conflicts with other activations/deactivations (as no other leaders exist) $G_2$ eventually becomes equal to $G_1$. Finally, note that the active edges of the matching between $V_1$ and $V_2$ are never deactivated but this is not a problem provided that $Q_{out}=\{r,r_a,r_d\}$ as each such edge $uv$ has $u$ in a state from $Q\bs Q_{out}$ and $v$ in a state from $Q_{out}$, that is it is not considered as part of the output.

Now for the running time we consider three phases: the \emph{matching formation}, the \emph{leader election}, and the \emph{unique-leader replication}. The matching formation phase begins from step 1 and ends when the last $q_0$ becomes $l$, i.e. when all nodes in $V_1$ have been matched to the nodes of $V_2$. It is not hard to see that the probability of the $i$th edge of the matching to be established (given $(i-1)$ established matches) is $p_i=[2(n/2-i)^2]/[n(n-1)]$ and the corresponding expectation is $\E[X_i]=1/p_i=\Theta(n^2/(n-i)^2)$. Then similarly to the coupon collector's application in the running time of Protocol \emph{Fast-Global-Line} in Theorem \ref{the:gline3} we have that the expected running time of this phase is $\E[X]=\Theta(n^2)$. An almost identical analysis yields that the expected running time of the leader election phase is also $\Theta(n^2)$. Thus it remains to estimate the time it takes for the unique leader to copy every edge of $E_1$. Given that the leader has marked the endpoints of a particular edge of $E_1$ then copying and restoring the state of the leader takes on average $\Theta(n^2)$ time (as a constant number of particular interactions must occur and each one occurs with probability $1/n^2$). Now we consider the time for copying as constant and try to estimate the time it takes for the leader to ``collect'' (i.e. visit and mark) all edges of $E_1$. Assume also that the leader is selected in every step to interact with one of its neighbors (the truth is that it is selected every $\Theta(n)$ steps on average). If $p_e$ is the probability that a specific edge $e$ is selected after two subsequent interactions then $p_e\simeq(1/n)(1/2)(1/n)(1/2)=\Theta(1/n^2)$, where $(1/n)(1/2)$ is the probability that the leader interacts with and decides to move on one endpoint of $e$ and $(1/n)(1/2)$ the probability that it then interacts with and decides to mark the other endpoint of $e$. Let $Y_i$ be the r.v. of the number of steps between the $(i-1)$th and $i$th edge collected and $p_i$ be the probability of a success in two consecutive steps of the $i$th epoch. Clearly, $p_i\simeq(n^2-i)/n^2$, $\E[Y_i]=1/p_i=n^2/(n^2-i)$, and $\E[Y]=\E[\sum_{i=0}^{n^2-1} Y_i]=n^2\sum_{i=0}^{n^2-1}1/(n^2-i)=n^2\sum_{i=1}^{n^2}1/i=\Theta(n^2\log n)$. Thus, provided that the leader always interacts and that every copying that it performs takes constant time, the expected time until the unique-leader replication phase ends is $\Theta(n^2\log n)$. Now, notice that on average it takes $\Theta(n)$ steps for the leader to interact and that in half of its interactions the leader performs a copying that takes $\Theta(n^2)$ steps to complete. That is, each of the above $\Theta(n^2\log n)$ steps is charged on average by $n$ and half of them are charged by $n^2$, i.e. half of the steps are charged by $\Theta(n)$ and the other half are charged by $n^2+n=\Theta(n^2)$. We conclude that the expected running time of the unique-leader replication phase is $\Theta(n^3\log n)+\Theta(n^4\log n)=\Theta(n^4\log n)$. This is clearly the dominating factor of the total running time of the protocol.
\qed
\end{proof}

Table \ref{tab:ulb} summarizes all upper and lower bounds that we established in Sections \ref{sec:global-line} and \ref{sec:basic-con}.

\begin{table}[h]
\normalsize
\setlength{\tabcolsep}{15pt}
\begin{center}
\begin{tabular}{  l  c  c c  }
  \hline
  Protocol & \# states & Expected Time & Lower Bound \\ \hline
  \emph{Simple-Global-Line} & 5 & $\Omega(n^4)$ and $O(n^5)$ & $\Omega(n^2)$ \\
  \emph{Intermediate-Global-Line} & 8 & $\Omega(n^3)$ and $O(n^4)$ & $\Omega(n^2)$ \\ 
  \emph{Fast-Global-Line} & 9 & $O(n^3)$ & $\Omega(n^2)$ \\
  \emph{Cycle-Cover} & 3 & $\Theta(n^2)$ (optimal) & $\Omega(n^2)$ \\
  \emph{Global-Star} & 2 (optimal) & $\Theta(n^2\log n)$ (optimal) & $\Omega(n^2\log n)$ \\ 
  \emph{Global-Ring} & 9 & & $\Omega(n^2)$ \\ 
  \emph{2RC} & 6 & & $\Omega(n\log n)$ \\ 
  \emph{$k$RC} & $2(k+1)$ &  & $\Omega(n\log n)$ \\ 
  \emph{$c$-Cliques} & $5c-3$ &  & $\Omega(n\log n)$ \\
  \emph{Leader-Replication} & 12 & $\Theta(n^4\log n)$ & \\ \hline
\end{tabular}
\end{center}
\caption{All upper and lower bounds established in Sections \ref{sec:global-line} and \ref{sec:basic-con}. \emph{Leader-Replication} is a randomized protocol thus it concerns class $\rem{PREL}$, while all other protocols do not rely on randomization thus they concern $\rem{REL}$.}
\label{tab:ulb}
\end{table}

\section{Generic Constructors}
\label{sec:gencon}

In this section, we ask whether there is a generic constructor capable of constructing a large class of networks. We answer this in the affirmative by presenting (i) constructors that simulate a Turing Machine (TM) and (ii) a constructor that simulates a distributed system with names and logarithmic local memories. Denote by $l$ the binary length of the input of a TM and by $n$ the size of a population. Note that Theorems \ref{the:gencon-half} to \ref{the:no-waste} construct a random graph in the useful space and that graph constitutes the input to a TM. Thus, if the useful space consists of $h$ nodes, the input to the TM has size $l=\Theta(h^2)$ and, as all of our constructors have $h=\Theta(n)$, in what follows it holds that $l=\Theta(n^2)$. Moreover, the TM can use space at most $O((n-h)^2)=O(n^2)$, where $n-h$ is the waste.

We now briefly describe the main idea behind all of our generic constructors that simulate a TM (see also Figure \ref{fig:main-loop}). Assume that we are given a decidable graph-language $L$ and we are asked to provide a NET that constructs $L$. The NET that we give works as follows:

\begin{enumerate}
 \item It constructs on $k$ of the nodes a network $G_1$ capable of simulating a TM and of constructing a random network on the remaining $n-k$ nodes. Let $V_1\subseteq V$ be the set of the $k$ nodes and $V_2=V\bs V_2$ the set of the remaining $n-k$ nodes. $G_1$ is usually a sufficiently long line or a bounded degree network as these networks can be operated as TMs. A line also serves as a measure of order as we can match a line of length $k$ with $k$ other nodes and by exploiting the ordering of the line we may achieve an ordering of the other nodes. 
 \item The NET exploits $G_1$ to construct a random network on $V_2$. The idea is to exploit the structure of $G_1$ so that it can perform a random coin tossing on each edge between nodes of $V_2$ exactly once. In this manner, it constructs a random network $G_2$ from $G_{n-k,1/2}$ on the nodes of $V_2$. It is worth noting that all networks of $G_{n-k,1/2}$ have an equal probability to occur and this results in an equiprobable constructor.
 \item The NET simulates on $G_1$ the TM that decides $L$ with $G_2$ as its input. The only constraint is that the space used by the TM should be at most the space that the constructor can allocate in $G_1$. If the TM rejects, then the protocol goes back to 2, that is it draws another random network and starts a new simulation. Otherwise, its output stabilizes to $G_2$.
\end{enumerate}

\begin{figure}[!hbtp]
\centering{
\includegraphics[width=0.53\textwidth]{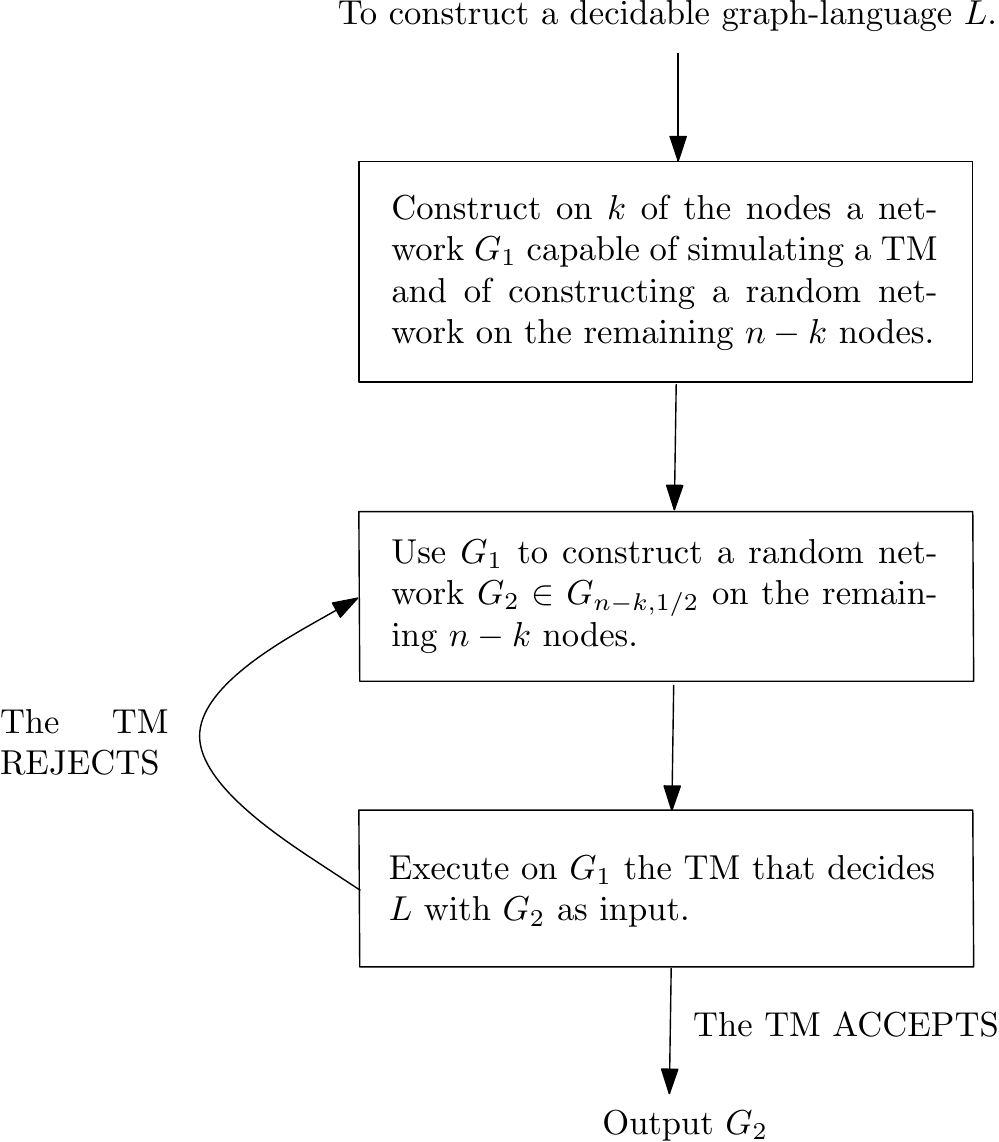}
}
\caption{The main mechanism used by all generic constructors in this section. The loop repeats until the TM accepts for the first time. When this occurs, the random graph $G_2$ constructed belongs to $L$ and thus the protocol may output $G_2$. Note that this is not a terminating step. The protocol just does not repeat the loop and thus its output forever remains stable to $G_2$.} \label{fig:main-loop}
\end{figure} 

\subsection{Linear Waste} 

\begin{theorem} [Linear Waste-Half] \label{the:gencon-half}
$\rem{DGS}(O(\sqrt{l}))\subseteq\rem{PREL}(\lfloor n/2\rfloor)$. In words, for every graph language $L$ that is decidable by a $O(\sqrt{l})$-space TM, there is a protocol that constructs $L$ equiprobably with useful space $\lfloor n/2\rfloor$.
\end{theorem}
\begin{proof}
We give a high-level description of the protocol, call it $\ca$. Let us begin by briefly presenting the main idea. Given a population of size $n$, $\ca$ partitions the population (apart from one node when $n$ is odd) into two equal sets $U$ and $D$ such that all nodes in $U$ are in state $q_u$, all nodes in $D$ are in state $q_d$ and each $u\in U$ is matched via an active edge to a $v\in D$, i.e. there is an active perfect matching between $U$ and $D$ (see Figure \ref{fig:gencon-overview}). By using the \emph{Simple-Global-Line} protocol (see Protocol \ref{prot:gline} in Section \ref{subsec:simple-global-line}) on the nodes of set $U$, $\ca$ constructs a spanning line in $U$ which has the endpoints in state $q_1$, the internal nodes in state $q_2$, and has additionally a unique leader on some node. We should mention that, though we use protocol Simple-Global-Line here as our reference, any protocol that constructs a spanning line would work. Given such a construction, $\ca$ organizes the line into a TM. The goal is for the TM to compute a graph from $L$ and construct it on the nodes of set $D$. To achieve this, the TM implements a binary counter ($\log n$ bits long) in its memory and uses it in order to uniquely identify the nodes of set $D$ according to their distance from one endpoint, say e.g. the left one. Whenever it wants to modify the state of edge $(i,j)$ of the network to be constructed, it marks by a special activating or deactivating state the $D$-nodes at distances $i$ and $j$ from the left endpoint, respectively. Then an interaction between two such marked $D$-nodes activates or deactivates, respectively, the edge between them. To compute a graph from $L$ equiprobably, the TM performs the following random experiment. It activates or deactivates each edge of $D$ equiprobably (i.e. each edge becomes active/inactive with probability $1/2$) and independently of the other edges. In this manner, it constructs a random graph $G$ in $D$ and all possible graphs have the same probability to occur. Then it simulates on input $G$ the TM that decides $L$ in $\sqrt{l}$ space to determine whether $G\in L$. Notice that the $n/2$ space of the simulator is sufficient to decide on an input graph encoded by an adjacency matrix of $(n/2)^2$ binary cells (which are the edges of $U$). If the TM rejects, then $G\notin L$ and the protocol repeats the random experiment to produce a new random graph $G^\prime$ and starts another simulation on input $G^\prime$ this time. When the TM accepts for the first time, the constructed random network belongs to $L$ and the protocol releases the constructed network by deactivating one after the other the active $(q_u,q_d)$ edges and at the same time updates the state of each $D$-node to a special $q_{out}$ state. Finally, we should point out that, whenever the global line protocol makes progress, all edges in $D$ are deactivated and the TM-configuration is \emph{reinitialized} to ensure that, when the final progress is made (resulting in the final line spanning $U$) the TM will be executed from the beginning on a correct configuration (free of residues from previous partial simulations).

\begin{figure}[!hbtp]
\centering{
\includegraphics[width=0.75\textwidth]{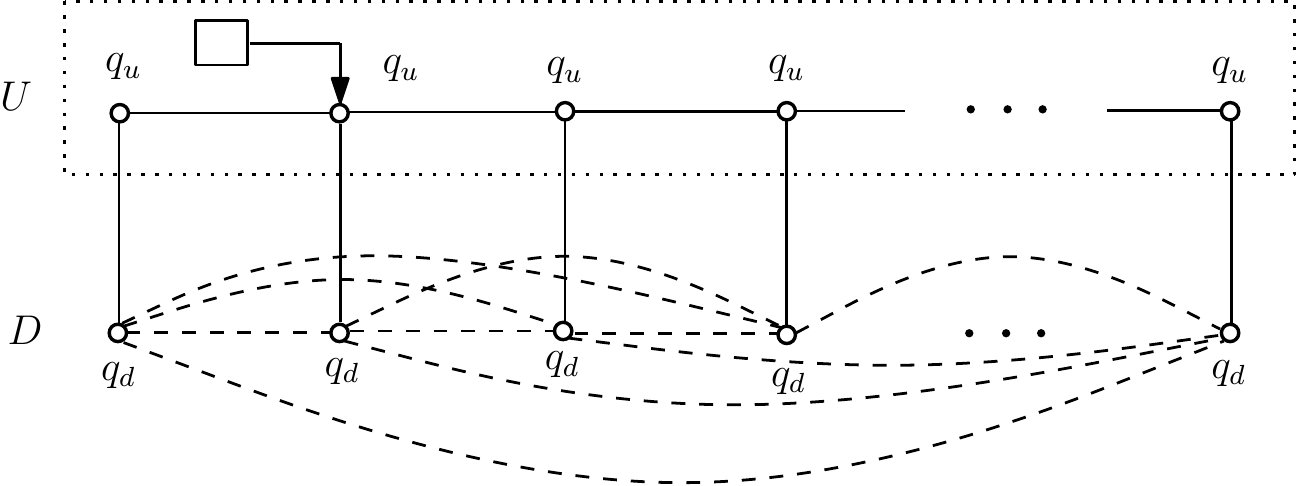}
}
\caption{The population partitioned into sets $U$ and $D$. The vertical active edges (solid) match the nodes of the two sets. The horizontal active edges between nodes in $U$ form a spanning line that is used to simulate a TM. The TM will construct the desired network on the nodes of set $D$ by activating the appropriate edges between them (dashed edges that are initially inactive).} \label{fig:gencon-overview}
\end{figure}  

We now proceed with a more detailed presentation of the various subroutines of the protocol.

\noindent\emph{Simulating the direction of the TM's head.} We begin by assuming that the spanning line has been constructed somehow (we defer for the end of the proof the actual mechanism of this construction), as in Figure \ref{fig:gencon-overview}, and that each node has three components $(c_1,c_2,c_3)$ in its state. $c_1$ is used to store the head of the TM, i.e. the actual state of the control of the TM; assume that initially the head lies on an arbitrary node, e.g. on the second one from the left as in Figure \ref{fig:gencon-overview}. $c_2$ is used to store the symbol written on each cell of the TM. $c_3$ is $l$, $r$, $t$ for ``left'', ``right'', and ``temporal'' respectively, or $\sqcup$ (for ``empty'') and we assume that initially the left endpoint is $l$, the right endpoint is $r$, and all internal nodes are $\sqcup$. As initially the head cannot have any sense of direction, it moves towards an arbitrary neighbor, say w.l.o.g. the right one, and leaves a $t$ on its previous position. The $t$ mark gives to the head a sense of direction on the line. The head can make progress towards one endpoint by just moving only towards the unmarked neighbor (avoiding the one marked by $t$). Once the head reaches the right endpoint for the first time, it starts moving towards the left endpoint by leaving $r$ marks on the way. Once it reaches the left endpoint it is ready to begin working as a TM. Now every time it wants to move to the right it moves onto the neighbor that is marked by $r$ while leaving an $l$ mark on its previous position. Similarly, to move to the left, it moves onto the $l$ neighbor and leaves an $r$ mark on its previous position. In this way, no matter what the position of the head will be, there will be always $l$ marks to its left and $r$ marks to its right, as in Figure \ref{fig:gencon-direction}, and the head can exploit them to move correctly. Additionally, we ensure that the endpoints are in special states, e.g. $l_e$ and $r_e$, to ensure that the head recognizes them in order to start moving to the opposite direction.  

\begin{figure}[!hbtp]
\centering{
\includegraphics[width=0.75\textwidth]{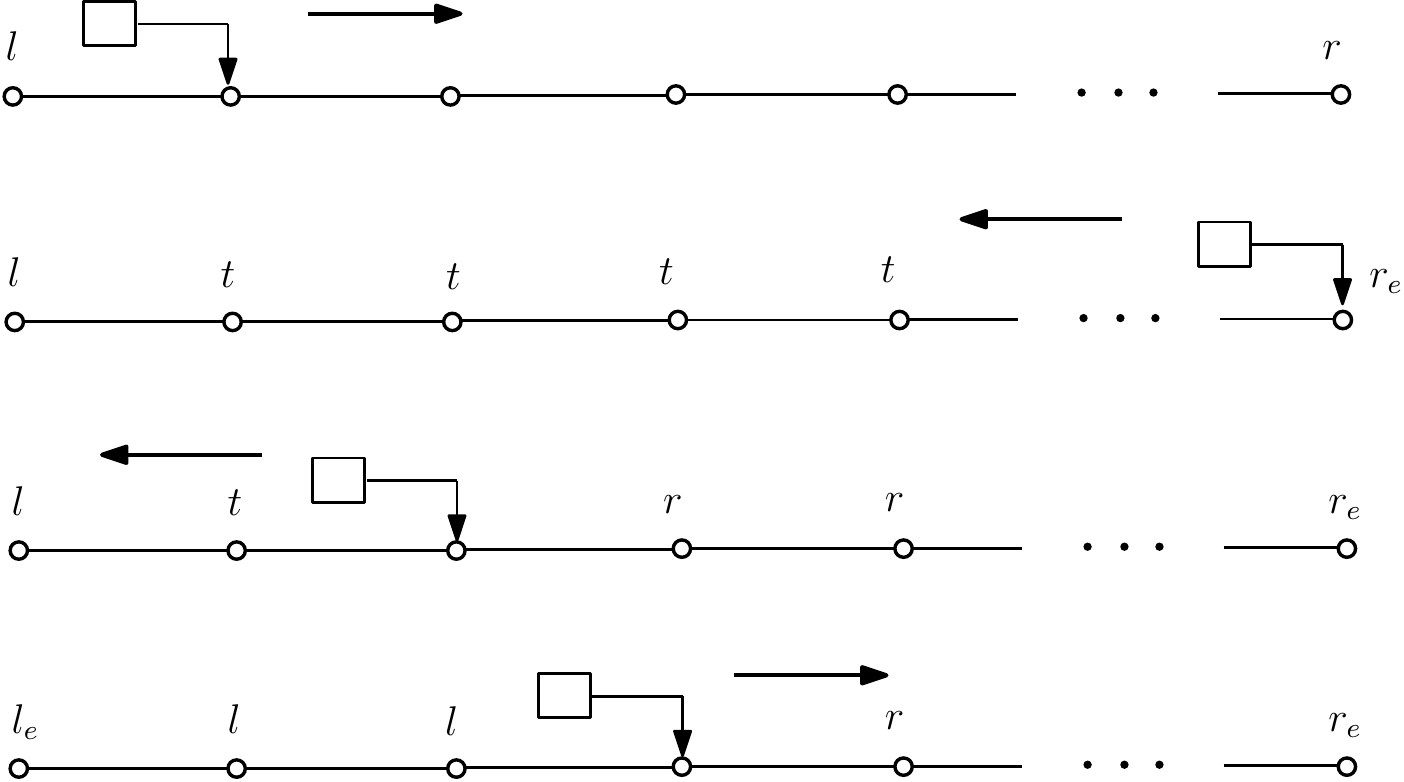}
}
\caption{The main idea of using $l$ and $r$ marks to simulate the movement of the head of a TM. The first three snapshots present the phase of the initialization of the marks where a temporal $t$ mark is used to move for the first time towards an endpoint. In the fourth snapshot, after the head has visited both endpoints, the $t$ marks have been removed and all nodes to the left of the head are marked $l$ while all nodes to the right are marked $r$. Additionally, the endpoints have special marks to ensure that the head recognizes them.} \label{fig:gencon-direction}
\end{figure}  
 
\noindent\emph{Reading and Writing on the edges of set $D$.} We now present the mechanism via which the TM reads or writes the state of an edge joining two $D$-nodes. The TM uniquely identifies a $D$-node by its distance from the left endpoint. To do this, it implements a binary counter on $\log n$ cells of its memory. Whenever it wants to read (write, resp.) the state of the edge joining the $D$-nodes $i$ and $j$, it sets the counter to $i$, places a special mark on the left endpoint, and repeatedly moves the mark one position to the right while decrementing the counter by one. When the counter becomes 0, it knows that the mark is over the $i$-th $U$-node. Now by exploiting the corresponding active vertical edge it may assign a special mark to the $i$-th $D$-node (Figure \ref{fig:gencon-rw-edges} provides an illustration). By setting the counter to $j$ and repeating the same process, another special mark may be assigned to the $j$-th $D$-node. Now the TM waits for an interaction to occur between the marked $D$-nodes $i$ and $j$. During that interaction edge $(i,j)$ is read (written, resp.) by the corresponding endpoints. Then, in case of a read (and similarly for a write), the TM reads the value of the edge that the endpoints detected, and in both cases unmarks both endpoints resetting them to their original states.

\begin{figure}[!hbtp]
\centering{
\includegraphics[width=0.75\textwidth]{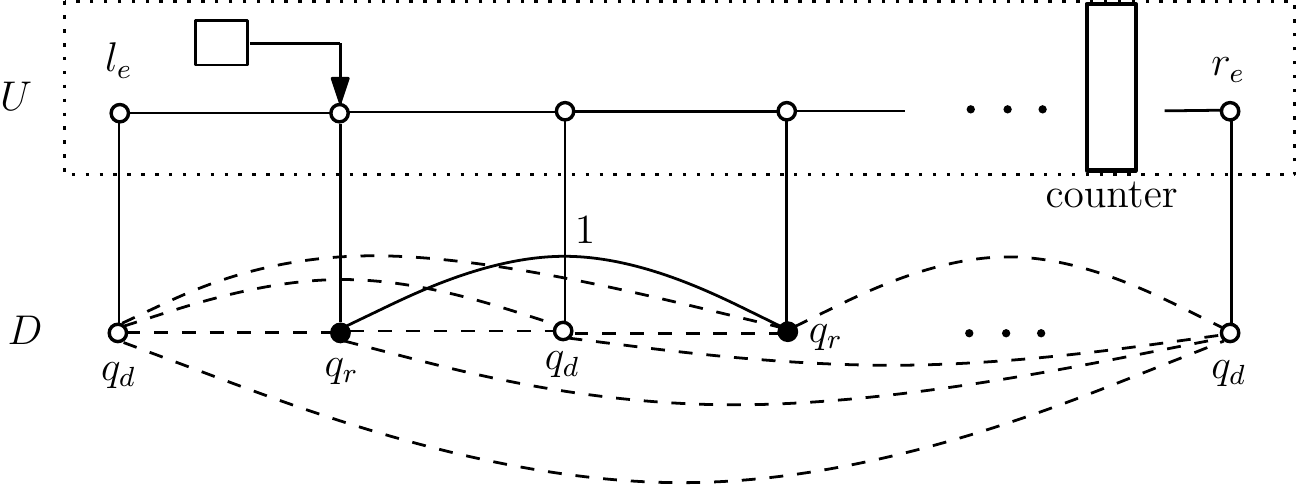}
}
\caption{By exploiting the implemented binary counter, the TM has managed to mark the desired nodes from set $D$, in this case the 2nd and the 4th ones counting from left, which are now in a special ``reading'' state $q_r$. An interaction between them will read the state of the edge joining them, which here happens to be an active one. Then the TM will read that value from one of these two nodes, in this case from the 2nd one. A write is implemented similarly.} \label{fig:gencon-rw-edges}
\end{figure}

\noindent\emph{Creating the input of the TM.} We now describe how the network construction works. As already stated, to simplify the description and in order to present an equiprobable constructor we have allowed nodes to toss a fair coin during their interaction. In particular, we allow transitions that with probability $1/2$ give one outcome and with probability $1/2$ another. Now before executing the simulation, the simulating protocol does the following. It visits one after the other the edges of set $D$ and on each one of them performs the following random experiment: with probability $1/2$ it activates the edge and with probability $1/2$ it deactivates it. The result of this random process is an equiprobable construction of a random graph. In particular, all possible graphs have the same probability to occur. Note that the protocol can detect when all random experiments have been performed because it can detect the endpoints of the spanning line. For example, to visit all edges one after the other we may: (i) place two marks on the left endpoint, let $1\leq i<j\leq n$ be their positions on the line, (ii) for all $1\leq i\leq n-1$, perform random experiments on all $i < j\leq n$ by starting the second mark from position $i+1$ and moving it each time one position to the right, (iii) the process stops when $i$ becomes $n$, i.e. when the first mark occupies the right endpoint (which can be detected). Thus we can safely compose the process that draws the random graph to the process that simulates the TM. Once the random graph has been drawn, the protocol starts the simulation of the TM. Notice that the input to the TM is the random graph that has been drawn on the edges of $D$ which provide an encoding equivalent to an adjacency matrix. There are $l=(n/2)^2$ edges and the simulator has available space $n/2=\sqrt{l}$, which is sufficient for the simulation of a $\sqrt{l}$-space TM. We now distinguish two cases, one for each possible outcome of the simulation.
\begin{enumerate}
 \item The TM rejects: In this case, the constructed random graph does not belong to $L$. The protocol repeats the random experiment, i.e. draws another random graph, and starts over the simulation on the new input.
 \item The TM accepts: The constructed graph belongs to $L$ and the protocol enters the Releasing phase (see below).
\end{enumerate}

\noindent\emph{Releasing.} When the TM accepts for the first time, the simulating protocol updates the head to a special finalizing state $f$. Now the head moves to the left endpoint and starts releasing one after the other the nodes of set $D$ by deactivating the vertical edges and updating the states of the released $D$-nodes to $q_{out}$. Now the network constructed over the nodes of set $D$ is free to move in the ``solution''.

It remains to resolve the following issue. In the beginning, we made the assumptions that the population has been partitioned into sets $U$ and $D$ and that a spanning line in $U$ has been constructed somehow. Though it is clear that the rule $(q_0,q_0,0)\rightarrow (q_u,q_d,1)$ can achieve the partitioning and that the \emph{Simple-Global-Line} protocol can construct a spanning line in $U$, it is not yet clear whether these processes can be safely composed to the simulating process. To get a feeling of the subtlety, consider the following situation. It may happen that a small subset $S$ of the nodes has been partitioned into sets $U^\prime$ and $D^\prime$ and that $U^\prime$ has been organized into a line spanning its nodes. If the nodes in $S$ do not communicate for a while to the rest of the network, then it is possible that a graph is constructed in $D^\prime$, which on one hand belongs to $L$ but on the other hand its order is much smaller than the desired $n/2$. To resolve this we introduce a reinitialization phase.\\

\vspace{-3pt}
\noindent\emph{Reinitialization.} A reinitialization phase is executed whenever a line on $U$-nodes expands (either by attracting free nodes or by merging with another line). At that point, the protocol ``makes the assumption'' that no further expansions will occur, restores the components of the simulation to their original values, ensures that each node in the updated set $U$ has a $D$-neighbor (as it is possible that some of them have released their neighbors), and initiates the drawing of a new random graph on the new set $D$. Though the assumption of the protocol may be wrong as long as further expansions of the line may occur, at some point the last expansion will occur and the assumption of the protocol will be correct. From that point on, the simulation will be reinitialized and executed for the last time on the correct sets $U$ and $D$. A final point that we should make clear is the following. During reinitialization we have two options: (i) block the line from further expansions until all components have been restored correctly and then unblock it again or (ii) leave it unblocked from the beginning. In the latter case, if another expansion occurs before completion of the previous reinitialization then another reinitialization will be triggered. However, if the two reinitialization processes ever meet then we can always kill one of them and restart a new single reinitialization process. Both options are correct and equivalent for our purposes. 
\qed
\end{proof}

We now show an interesting trade-off between the space of the simulated TM and the order of the constructed network. In particular, we prove that if the constructed network is required to occupy $1/3$ instead of half of the nodes, then the available space of the TM-constructor dramatically increases to $O(n^2)$ from $O(n)$.

\begin{theorem} [Linear Waste-Two Thirds]
$\rem{DGS}(l+O(\sqrt{l}))\subseteq\rem{PREL}(\lfloor n/3\rfloor)$. In words, for every graph language $L$ that is decidable by a $(l+O(\sqrt{l}))$-space TM, there is a protocol that constructs $L$ equiprobably with useful space $\lfloor n/3\rfloor$.
\end{theorem}
\begin{proof}
The idea is to partition the population into three equal sets $U$, $D$, and $M$ instead of the two sets of Theorem \ref{the:gencon-half}. The purpose of sets $U$ and $D$ is more or less as in Theorem \ref{the:gencon-half}. The purpose of the additional set $M$ is to constitute a $\Theta(n^2)$ memory for the TM to be simulated. The goal is to exploit the $(n/3)(n/3-1)/2$ edges of set $M$ as the binary cells of the simulated TM (see Figure \ref{fig:gencon-one-third}). The set $U$ now, instead of executing the simulation on its own nodes, uses for that purpose the edges of set $M$. Reading and writing on the edges of set $M$ is performed in precisely the same way as reading/writing the edges of set $D$ (described in Theorem \ref{the:gencon-half}).

\begin{figure}[!hbtp]
\centering{
\includegraphics[width=0.75\textwidth]{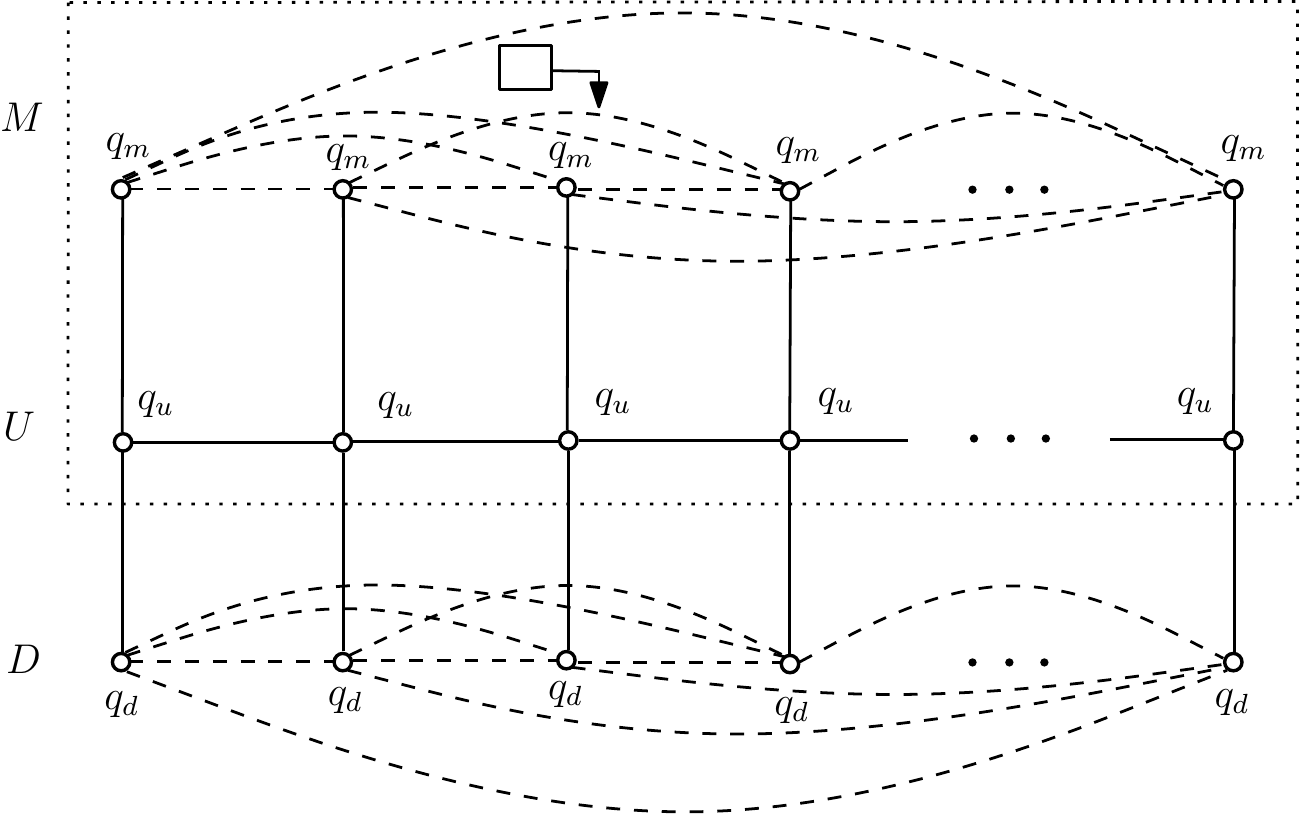}
}
\caption{A partitioning into three equal sets $U$, $D$, and $M$. The line of set $U$ plays the role of an ordering that will be exploited both by the random graph drawing process and by the TM-simulation. The line of set $U$ instead of using its $\Theta(n)$ memory as the memory of the TM it now uses the $\Theta(n^2)$ memory of set $M$ for this purpose. Set $D$ is again the useful space on which the output-network will be constructed. Sets $U$ and $M$ constitute the waste.} \label{fig:gencon-one-third}
\end{figure}

As everything works in precisely the same way as in Theorem \ref{the:gencon-half}, we only present the subroutine that constructs the $(U,D,M)$ partitioning.

\noindent\emph{Constructing the $(U,D,M)$ partitioning.} The rules that guarantee the desired partitioning into the three sets are:
\begin{align*}
(q_0,q_0,0)&\ra (q_u^\prime,q_d,1)\\
(q_u^\prime,q_0,0)&\ra (q_u,q_m,1)\\
(q_u^\prime,q_u^\prime,0)&\ra (q_u,q_m^\prime,1)\\
(q_m^\prime,q_d,1)&\ra (q_m,q_0,0)
\end{align*}
The idea is to consider a $U$-node as unsatisfied as long as it has not managed to obtain a $q_m$ neighbor. The unsatisfied state of a $U$-node is $q_u^\prime$. If a $q_u^\prime$ meets a $q_0$ then it makes that $q_0$ its $q_m$ neighbor and becomes satisfied. Note that it is possible that at some point the population may only consist of $q_u^\prime$ nodes matched to $D$-nodes which is not a desired outcome. For this reason, we have allowed $q_u^\prime$ nodes to be capable of making other $q_u^\prime$ nodes their $q_m$ neighbors. That is, when two $q_u^\prime$ nodes interact, one of them becomes satisfied, the other becomes $q_m^\prime$, and the edge joining them becomes active. A $q_m^\prime$ just waits to meet its active connection to a $D$-node, deactivates it, isolates the $D$-node by making it $q_0$ again, and becomes $q_m$. For an illustration, see Figure \ref{fig:gencon-three-sets}. Then, for the construction of the line spanning $U$, we only allow satisfied $U$-nodes to participate to the construction. As a satisfied $U$-node never becomes unsatisfied again, this choice is safe.

\begin{figure}[!hbtp]
\centering{
\includegraphics[width=0.85\textwidth]{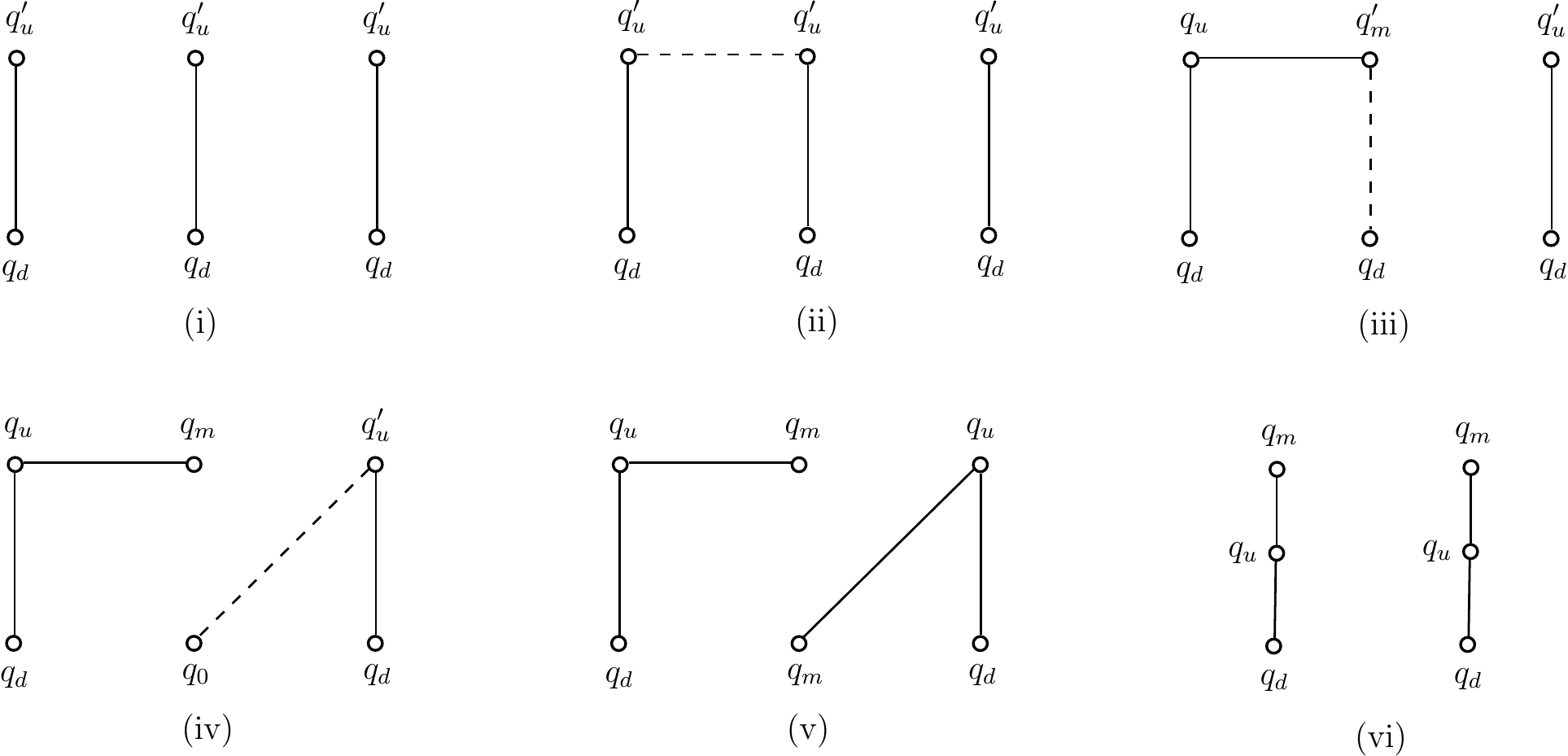}
}
\caption{An example construction of a $(U,D,M)$ partitioning.} \label{fig:gencon-three-sets}
\end{figure}
\qed
\end{proof}

\subsection{Logarithmic Waste}

We now relax our requirement for simulation space in order to reduce the waste (which, in both of the previous two theorems, was of the order of $n$).

\begin{theorem} [Logarithmic Waste] \label{the:logar}
$\rem{DGS}(O(\log l))\subseteq\rem{PREL}(n-\log n)$. In words, for every graph language $L$ that is decidable in logarithmic space, there is a protocol that constructs $L$ equiprobably with useful space $n-\log n$.
\end{theorem}
\begin{proof}
We give the main idea. The protocol first constructs a spanning line. Let us for now assume that the spanning line has been somehow constructed by the protocol. Then the protocol exploits the line to count the number of nodes in the network. We may assume that counting is performed in the rightmost cells of the line. The head visits one after the other the nodes from left to right and for each next move it increments the binary counter by one. When the head reaches the right endpoint, counting stops and the binary counter will have occupied approximately $\log n$ nodes (in fact, the rightmost $\log n$ nodes). Now the protocol releases the counter without altering its line structure and additionally makes all remaining $n-\log n$ nodes isolated by resetting their states and deactivating the edges between them. From now on, we may assume w.l.o.g. that there is a line of $\log n$ nodes with a unique leader and with a distributed variable containing a very good estimate of the number of isolated nodes (for this, we just compute in the logarithmic memory $n-\log n$, where $n$ was already stored in binary and $\log n$ is the number of cells of the memory; another way to achieve this is to stop counting when the head - moving from left to right - reaches the first, i.e. leftmost, cell occupied by the counter). All nodes of the memory are in a special $m$ state while all remaining nodes are in some other state, e.g. $f$, so the two sets are distinguishable. Next the leader starts a random experiment in order to construct a random graph on the free nodes as follows. It picks the first free node that it sees, call it $u_1$, activates the edge between them and informs it to start tossing coins on each one of the edges joining it to other free nodes. Whenever $u_1$ tosses a coin on a new edge, it marks the corresponding node to avoid it in the future and informs the leader to decrement its ($n-\log n$)-counter by 1. When the counter becomes 0, $u_1$ has tossed coins on all its edges, by a similar counting process it removes all marks from the other free nodes, and remains marked so that the leader avoids picking it again in the future. Then the leader moves to some other free node $u_2$, repeating more or less the same process. At the same time the leader decrements another ($n-\log n$)-counter by one to know when all free $u_i$s have been picked. In this manner, a random graph is drawn equiprobably on the set of free nodes. Next, the leader simulates a logarithmic TM in its memory trying to decide whether the random graph belongs or not to a given language $L$. If yes, then we are done. If not, then the TM just repeats the random experiment and restarts the simulation.\\

\vspace{-3pt}
\noindent\emph{Reinitialization.} Clearly, the protocol cannot know when the line that it was initially trying to construct has become spanning. Due to this, after every expansion of the line it assumes that the line has become spanning and starts counting. It is clear that every counting process leads to the formation of a small line with a leader (of length logarithmic in the length of the original line) and several free nodes. The small line and its leader are kept forever by the simulation process. This implies that if there are more than one such lines, they will eventually interact and detect that their original line was not spanning. At that point, the interacting lines may merge to form a new line. It is clear that the only stable case is the one in which the original line was spanning and this will eventually occur. 
\qed
\end{proof}

\subsection{No Waste}

Going one step further we prove that a large class of graph-families can be constructed with no waste.

\begin{theorem} [No Waste] \label{the:no-waste}
Let $L$ be a graph language such that:  (i) there exists a natural number $d$ s.t. for all $G\in L$ there is a subgraph $G^\prime$ of $G$ of logarithmic order s.t. either $G^\prime$ or its complement is connected and has degree upper bounded by $d$ and (ii) $L$ is decidable in logarithmic space. Then $L\in\rem{PREL}(n)$, i.e. there is a protocol that constructs $L$ equiprobably with useful space $n$.
\end{theorem}
\begin{proof}
We give again the main idea. As in Theorem \ref{the:logar}, the protocol first constructs a spanning line used to separate a subpopulation $S$ of $V_I$ of size approximately $\log n$. Before deactivating the line of $T=V_I\bs S$ of length $n-\log n$ the protocol first exploits it to construct a random graph in $S$ of active or inactive degree (choosing randomly between these) upper bounded by $d$ (note that $d$ is finite and thus it is known in advance by the protocol). Then the line of $T$ organizes the bounded-degree graph of $S$ into a TM $M$ (which is feasible due to the fact that the degree is bounded; see Theorem 7 of \cite{AACFJP05}) of logarithmic space with a unique leader on some node. Next $M$ draws (more or less as in Theorem \ref{the:logar}) a random graph on the edges of $E_I\bs E[S]$, i.e. on all edges apart from those between the nodes of $S$ (to prevent destroying the structure of the TM). Note that, in order for the TM to be able to distinguish the nodes of $S$, the protocol has all these nodes in a special state that is not present in $T$. Observe now that, in this manner, the protocol has constructed on $V_I$ a random graph from those having a connected subgraph of logarithmic order and degree upper bounded by $d$. It remains to verify whether the one constructed indeed belongs to $L$. To do this, $M$ simulates the TM $N$ that decides $L$ in logarithmic space. If $N$ rejects then $M$ builds another line in $T$ that repeats the whole process, i.e. draws a new random graph in $S$ and so on. When $N$ first accepts, the protocol sleeps (in the sense that it does not terminate but does not alter edges anymore either).  
\qed
\end{proof}

\begin{remark}
If the graph-property $L$ (in any of the above results) happens to occur with probability at least $1/f(n)$, where $f(n)$ is polynomial on $n$, in the $G_{n,1/2}$ random graph model, then its corresponding generic constructor runs in polynomial expected time. Connectivity is such an example as every $G\in G_{n,\Theta(\log n/n)}$ is almost surely connected and the same holds for every $G\in G_{n,1/2}$ (hamiltonicity is another example).
\end{remark}

\begin{remark}
All the above generic results have been proved for $\rem{PREL}$. Note that in $\rem{REL}$ we can again construct a sufficiently long line (as our protocols for global line are in $\rem{REL}$) and exploit it as a space-bounded TM of the following sort: on input $g(n)$ (i.e. the useful space) the TM outputs a graph of order $g(n)$. By exploiting such graph-constructing TMs we can again construct a possibly large class of networks without giving to our protocols access to randomization.
\end{remark}

\subsection{Constructing and Simulating Supernodes with Logarithmic Memories}

We now show that a population consisting of $n$ nodes can be partitioned into $k$ \emph{supernodes} each consisting of $\log k$ nodes, for the largest such $k$. The internal structure of each supernode is a line, thus it can be operated as a TM of memory logarithmic in the total number of supernodes. This amount of storage is sufficient for the supernodes to obtain unique names and exploit their names and their internal storage to realize nontrivial constructions. We are interested in the networks that can be constructed at the supernode abstraction layer. The following theorem establishes that such a construction is feasible and presents a network constructor that achieves it.  

\begin{theorem} [Partitioning into Supernodes] \label{the:supernodes}
For every network $G$ that can be constructed by $k$ nodes having local memories $\lceil \log k\rceil$ and unique names there is a NET that constructs $G$ on $n=k\lceil\log k\rceil$ nodes.  
\end{theorem}
\begin{proof}
We present a NET $\ca$ that when executed on $n$ nodes it is guaranteed to organize the nodes into $k$ lines of length $\lceil\log k\rceil$ each for the maximum $k$ for which $k\lceil\log k\rceil \leq n$. We assume a unique pre-elected leader in the initial configuration of the system and we will soon show how to drop this requirement. Assume also for simplicity that $n\geq 8$ (this is again not necessary). The protocol operates in phases. Variable $j$ denotes the current phase number, $r$ denotes the number of new lines that should be constructed in the current phase, and $a$ is a line counter. We assume that the leader has somehow already created 4 lines of length 2 each (note that here we count the length of a line in terms of its nodes). One of them is the leader's line. Also the left endpoint of the leader's line is directly connected  to the left endpoints of the other 3 lines. In fact, all these assumptions are trivial to achieve. Initially $j\leftarrow 2$. All variables are stored by the leader in the distributed memory of its line.
\begin{itemize}
 \item A new phase starts when the leader manages to increase by one the length of its line by attaching an isolated node its right endpoint. When this occurs, the leader sets $j\leftarrow j+1$, $r\leftarrow 2^{j-1}$, and $a\leftarrow 2$. A phase is divided into two subphases: the \emph{Increment existing lines} subphase and the \emph{Create new lines} subphase.
  \begin{itemize}
  \item \emph{Increment existing lines}: Initially, all existing lines, excluding the leader's line, are marked as \emph{unvisited}. While $a\leq r$ the leader visits an unvisited line and tries to increment its length by one by attaching an isolated node to its right endpoint. When it succeeds, it marks the line as \emph{visited}, sets $a\leftarrow a+1$ and returns to its own line. When this subphase ends all existing lines have length $j$. Then the leader sets $a\leftarrow 1$ and the \emph{Create new lines} subphase begins.
  \item \emph{Create new lines}: While $a\leq r$ the leader becomes connected to an isolated node, it marks that node as the left endpoint of the new line and then starts creating the new line node-by-node, by attaching isolated nodes to its right. It stops increasing the length of the new line when it becomes equal to the length of its own line. This can be easily implemented by a mark on the leader's line that moves one step to the right every time the length of the new line increases by one. The new line has the right length when the mark reaches the right endpoint of the leader's line. When this subphase ends there is a total of $2r=2^j$ lines of length $j$ each and the leader is directly connected to the left endpoint of each one of them. Then the leader waits again to increase its own length by one and when this occurs a new phase begins. 
  \end{itemize}
\end{itemize}

\emph{Naming.} We now show that it is not hard to keep the constructed lines named (in fact there are various strategies for achieving this). Initially, the leader has 4 lines of length 2 each and we may assume that these are uniquely named $0,1,2,3$ in binary, that is every line has its name stored in its own memory. During a phase, the leader keeps a variable \emph{cname} storing the current name to be assigned, initially 0. Whenever the leader increases the length of an existing line (during the \emph{increment} subphase) or creates a new line (during the \emph{create} subphase) it assigns to it \emph{cname} in binary and sets $cname\leftarrow cname+1$. Clearly, at the end of phase $j$ the lines are uniquely named $0,1,\ldots,2^{j}-1$.\\

\vspace{-3pt}
\emph{Electing the Leader.} We now show how to circumvent the problem of not having initially a unique pre-lected leader. In fact, as we soon discuss, the solution we develop may serve as a generic technique for simulating protocols that assume a pre-elected leader. Initially all nodes are leaders in state $l_0$. Rule $(l_0,l_0,0)\ra (l_0,q_0,0)$ eliminates one leader when two of them interact. These leaders start executing the above protocol by attaching $q_0$ nodes to their construction. By the time a $l_0$ leader attaches the first isolated node to its construction it remains leader but changes its state to $l$. Each leader executes the protocol on its own constructed component until it meets another leader. One this occurs, one of the two $l$s becomes $w$. The goal of a $w$ leader is to revert its whole component to a set of isolated nodes in a special sleeping state $s$ (itself inclusive). A sleeping node can like a $q_0$ node be attached to lines by leaders but unlike a $q_0$ node it cannot participate in the creation of another leader. First of all note that it can easily revert a single line by beginning from the right endpoint and releasing one after the other the nodes until it reaches the left endpoint. In fact, the generic idea (that works for other constructions as well) is that in order to release a node it suffices to know its degree. Then the only possible difficulty in our case is the fact that the left endpoint of the leader's line may be connected to a non-constant number of other endpoints. To resolve this, the leader exploits the fact that it can count in its line's memory the number of lines. When the reversion process begins, the leader knows the number of lines, that is it knows also the degree of the left endpoint of its line. Whenever it reverts another line it decreases the counter by one. So, when the counter becomes equal to 1, it knows that the only remaining line is its own line, thus it knows that when it comes to release the last two nodes of its own line (i.e. during the interaction between the left endpoint and the other remaining node of the line) it should make both sleeping as there is no other reversion to be performed. This is quite important as it guarantees that reverting does not introduce waste. Note that if the reversion process could not determine its completion then every such reversion would result in a node remaining forever in state $w$. Such zombie $w$s cannot be exploited by other leaders in their constructions as allowing a leader to attach a $w$ would introduce conflicts between constructing and reverting processes.\\

\vspace{-3pt}
\emph{Reinitialization.} Note that the simulated protocol that constructs $G$ assuming memories and names must be executed from the beginning because protocol $\ca$ that gives the organization into lines is not terminating thus the two protocols must be composed in parallel. It suffices to have every line remember the number of active edges that it has to other lines. Then, whenever a new phase begins (implying that what has been constructed so far by the simulated protocol is not valid), each line deactivates one after the other all those edges and starts over the simulation.\\

\vspace{-3pt}
The only drawback is that the above protocol retains forever the connections between the left endpoint of the leader's line and the left endpoints of the other lines. However, if we agree that the output-network of the protocol is the one induced by the active edges joining the right endpoints of lines then this is not an issue. Additionally, it should not be that hard to circumvent this subtlety by having the leader periodically release the constructed lines and reattracting them only in case it manages to increase the length of one of them.
\qed
\end{proof}

Many network construction problems are substantially simplified given the supernodes with names and memories. For a simple example, consider the problem of partitioning the nodes into triangles. This construction is quite hard to achieve in the original setting without a leader, however, given the supernodes it becomes trivial. Each supernode with id $i$ checks whether its id is a multiple of 3 and, if it is, it connects to id $(i + 2)$, otherwise it connects to id $(i - 1)$. This is a totally parallel and thus a very efficient solution.

Finally, the above approach introduces the idea of constructing disjoint stable structures and then looking at those structures from a higher level and considering them as units (supernodes). It is then challenging, interesting, and valuable to understand how these units behave, what is the dependence of their behavior to their internal structure and configuration, what is the outcome of an interaction between two such units, and what are their constructive capabilities. In fact, one can imagine a whole such hierarchy of layers were nodes self-assemble into supernodes, supernodes self-assemble into supersupernodes, and so on. Formalizing this hierarchy is a very promising and totaly open research direction.

\section{Conclusions and Further Research}
\label{sec:conclusions}

There are many open problems related to the findings of the present work. Though our universal constructors show that a large class of networks is in principle constructible, they do not imply neither the simplest nor the most efficient protocol for each single network in the class. To this end, we have provided direct constructors for some of the most basic networks, but there are still many other constructions to be investigated like grids or planar graphs. Moreover, a look at Table \ref{tab:ulb} makes it evident that there is even more work to be done towards the probabilistic analysis of protocols and in particular towards the establishment of tight bounds. Of special interest is the spanning line problem as it is a key component of universal construction. All of our attempts to give a protocol asymptotically faster than $O(n^3)$ have failed. Observe that with a preelected leader in state $l$ and all edges initially inactive, the straightforward protocol $(l,q_0,0)\ra (q_1,l,1)$ produces a stable spanning line in an expected number of $\Theta(n^2\log n)$ steps (follows from the meet everybody fundamental process). Moreover, by a one-to-one elimination we can elect a unique leader in an expected number of $\Theta(n^2)$ steps. If we could safely compose these two protocols, then we would obtain a $\Theta(n^2\log n)$ constructor which is almost optimal as our present best lower bound for the spanning line is $\Omega(n^2)$. The problem is that the protocol cannot detect when the leader-election phase has completed, thus it has to activate edges while still having more than one leaders but this gives an overhead for either merging the constructed disjoint lines or deactivating some wrong connections. A possible solution could be to consider Monte Carlo protocols that may err with some small probability, e.g. a protocol that would try somehow to estimate when w.h.p. the leader-election phase completes and only then start the line construction phase. 

One of the problems that we considered in this work, was the problem of constructing any $k$-regular network. Note that this is a quite different problem than the problem of constructing a specific $k$-regular network. For example, given a population of 10 processes is there a protocol that stabilizes to the Petersen graph? In general, it is worth considering non-uniform protocols that when executed on the correct number of nodes are required to construct a unique network like e.g. the cubical graph or the Wagner graph on 8 processes. 

Another very intriguing issue has to do with the size of network constructors. In particular, we would like to know whether there is some generic lower bound on the size of all constructors, to give problem-specific lower bounds, and to formalize the apparent relationship between the size and the running time of a protocol. Is there some sort of hierarchy showing that with more states we can produce faster protocols (until optimality is obtained)? To this end observe that neither the maximum degree nor the number of different degrees of the target-network are lower bounds on the number of states required to construct the network. For the former, it is not hard to show that $\Theta(x)$ states suffice to make a node obtain $2^x$ neighbors (stably). The idea is to have a node initially obtain 2 neighbors and then repeatedly double their number. For the latter, one can show that $\Theta(x)$ states suffice to have $2^x$ nodes with different degrees (stably) and in particular for all $i\in \{1,..., 2^x\}$ we obtain a node with degree $i$. The idea is to mark a set of $2^x$ nodes as before and construct a line spanning these nodes. Then the protocol assigns to the $i$th node of the line, counting e.g. from the left endpoint, $i$ neighbors. This can be done by using only a constant number of states. The head begins from the left endpoint and moves step-by-step on the line towards $u$. For every step it takes it assigns to $u$ a new neighbor and stops when it reaches $u$. In this manner, it assigns to $u$ a number of neighbors equal to its distance from the endpoint without having to explicitly count the distance. Is there some other property of the target-network that determines the number of states that have to be used?

All of our generic constructors, produce a random graph in the useful space and then simulate a space bounded TM in the waste to determine whether the random graph belongs to the language and therefore whether it should be accepted. An interesting open problem is to characterize the class $\rem{REL}$ in which protocols do not have access to (internal) randomness. In this case, it seems useless to simulate deciding TMs and the focus should be on \emph{graph-constructing TMs} which are much less understood. Such a characterization would also help us appreciate whether randomization increases the constructive power of the model. It is worth noting that our results on universal construction indicate that the constructive power increases as a function of the available waste. A complete characterization of this dependence would be of special value. 

There is also a practically unlimited set of variations of the proposed model that is worth considering. We mention a few of them. As already discussed, in this work we have considered a model of network construction with as minimal assumptions as possible to serve as a simple and clear starting point for more applied models to be defined. We now introduce such a model which seems to be of particular interest. Assume that every node is equipped with a predefined number of ports at specific positions of its ``body''. For example, in the 2-dimensional case these could be ``North'', ``South'', ``East'', ``West'' having the obvious angles between them. Nodes interact via their ports and they can detect which of their ports are used in an interaction. Moreover, when a connection is activated, it is always activated at a predetermined distance (i.e. all connections have the same length $d$) and it is always a straight line respecting the angles between itself and the (potentially active) lines of the other ports of the same node. Such a model (and possible variations of it, depending on the assumed hardware) seems particularly suitable for studying/designing very simple and local distributed protocols that are capable of constructing stable geometric objects (even in three dimensions), like squares, cubes, or more complex polyhedra, without any mobility-control mechanism. Moreover, an immediate extension of our model is to allow the connections to have more than just the two states that we considered in this work. Recall that, whenever we had to analyze the running time of a protocol, we did it under the uniform random scheduler, mainly because we wanted to keep this first model of network construction as simple as possible and because of its correspondence to a well-mixed solution. However, there are many other natural probabilistic scheduling models to be considered which would probably require different algorithmic developments and techniques to achieve efficiency. It is also natural to consider a variant in which connected nodes communicate much faster (even in synchronous rounds) than disconnected nodes. Moreover, it would be interesting to consider a model of network construction in which the behavior of a node depends on some input from the environment (this would allow the consideration of codes that exhibit different behaviors in different environments). Interesting and natural seems also the model in which a connected component has access to a \emph{self-bit} indicating whether a given interaction involves two nodes of the same component or not. It is not yet clear whether this extra assumption increases the constructive power of the model but it is clear that it substantially simplifies the description of several protocols. Also, of its own value would be to depart from cooperative models and consider an antagonistic scenario in which different sets of nodes try to construct different networks (by deterministic codes and not game-theoretic assumptions involving incentives). It would be interesting to discover cases in which the antagonism leads to unexpected stable formations. 

Finally, a very valuable and challenging interdisciplinary goal is to further investigate and formalize the apparent applicability of the model proposed here (and potential variations of it) in physical and chemical (possibly biological) processes. As already stated, we envision that a potential usefulness of such models is to unveil the algorithmic properties underlying the structure/network formation capabilities of natural processes.\\

\noindent \textbf{Acknowledgements.} We would like to thank Leslie Ann Goldberg for bringing to our attention the importance of constructing regular networks and also the reviewers of this work and some previous versions of it whose comments have helped us to improve our work substantially.

\bibliographystyle{alpha-abr}
\bibliography{podc14-full}

\newcommand{\etalchar}[1]{$^{#1}$}
\begin{thebibliography}{VMC{\etalchar{+}}12}

\bibitem[AAC{\etalchar{+}}05]{AACFJP05}
D.~Angluin, J.~Aspnes, M.~Chan, M.~J. Fischer, H.~Jiang, and R.~Peralta.
\newblock Stably computable properties of network graphs.
\newblock In {\em 1st IEEE International Conference on Distributed Computing in
  Sensor Systems (DCOSS)}, volume 3560 of {\em LNCS}, pages 63--74.
  Springer-Verlag, June 2005.

\bibitem[AAD{\etalchar{+}}06]{AADFP06}
D.~Angluin, J.~Aspnes, Z.~Diamadi, M.~J. Fischer, and R.~Peralta.
\newblock Computation in networks of passively mobile finite-state sensors.
\newblock {\em Distributed Computing}, pages 235--253, March 2006.

\bibitem[AAER07]{AAER07}
D.~Angluin, J.~Aspnes, D.~Eisenstat, and E.~Ruppert.
\newblock The computational power of population protocols.
\newblock {\em Distributed Computing}, 20[4]:279--304, November 2007.

\bibitem[Adl94]{Ad94}
L.~M. Adleman.
\newblock Molecular computation of solutions to combinatorial problems.
\newblock {\em Science}, 266[11]:1021--1024, November 1994.

\bibitem[Ang80]{An80}
D.~Angluin.
\newblock Local and global properties in networks of processors.
\newblock In {\em Proceedings of the 12th annual ACM symposium on Theory of
  computing (STOC)}, pages 82--93. ACM, 1980.

\bibitem[BA99]{BA99}
A.-L. Barab{\'a}si and R.~Albert.
\newblock Emergence of scaling in random networks.
\newblock {\em Science}, 286[5439]:509--512, 1999.

\bibitem[BBCK10]{BBCK10}
J.~Beauquier, J.~Burman, J.~Clement, and S.~Kutten.
\newblock On utilizing speed in networks of mobile agents.
\newblock In {\em Proceedings of the 29th ACM SIGACT-SIGOPS symposium on
  Principles of distributed computing (PODC)}, pages 305--314. ACM, 2010.

\bibitem[BEK{\etalchar{+}}13]{BEK13}
L.~Blume, D.~Easley, J.~Kleinberg, R.~Kleinberg, and {\'E}.~Tardos.
\newblock Network formation in the presence of contagious risk.
\newblock {\em ACM Transactions on Economics and Computation}, 1[2]:6, 2013.

\bibitem[Bol01]{Bo01}
B.~Bollob{\'a}s.
\newblock {\em Random graphs}, volume~73.
\newblock Cambridge university press, 2001.

\bibitem[BPS{\etalchar{+}}10]{BPSPF10}
A.~Bandyopadhyay, R.~Pati, S.~Sahu, F.~Peper, and D.~Fujita.
\newblock Massively parallel computing on an organic molecular layer.
\newblock {\em Nature Physics}, 6[5]:369--375, 2010.

\bibitem[CMN{\etalchar{+}}11]{MNPS11}
I.~Chatzigiannakis, O.~Michail, S.~Nikolaou, A.~Pavlogiannis, and P.~G.
  Spirakis.
\newblock Passively mobile communicating machines that use restricted space.
\newblock {\em Theoretical Computer Science}, 412[46]:6469--6483, October 2011.

\bibitem[DBC12]{DBC12}
S.~M. Douglas, I.~Bachelet, and G.~M. Church.
\newblock A logic-gated nanorobot for targeted transport of molecular payloads.
\newblock {\em Science}, 335[6070]:831--834, 2012.

\bibitem[DFSY10]{DFSY10}
S.~Das, P.~Flocchini, N.~Santoro, and M.~Yamashita.
\newblock On the computational power of oblivious robots: forming a series of
  geometric patterns.
\newblock In {\em Proceedings of the 29th ACM SIGACT-SIGOPS symposium on
  Principles of distributed computing (PODC)}, pages 267--276, 2010.

\bibitem[DGRS13]{DGRS13}
S.~Dolev, R.~Gmyr, A.~W. Richa, and C.~Scheideler.
\newblock Ameba-inspired self-organizing particle systems.
\newblock {\em arXiv preprint arXiv:1307.4259}, 2013.

\bibitem[Dot12]{Do12}
D.~Doty.
\newblock Theory of algorithmic self-assembly.
\newblock {\em Communications of the ACM}, 55:78--88, 2012.

\bibitem[Dot14]{Do14}
D.~Doty.
\newblock Timing in chemical reaction networks.
\newblock In {\em Proc. of the 25th Annual ACM-SIAM Symp. on Discrete
  Algorithms (SODA)}, pages 772--784, 2014.

\bibitem[ER59]{ER59}
P.~Erd{\H{o}}s and A.~R{\'e}nyi.
\newblock On random graphs.
\newblock {\em Publicationes Mathematicae Debrecen}, 6:290--297, 1959.

\bibitem[Fel68]{Fe68}
W.~Feller.
\newblock {\em An Introduction to Probability Theory and Its Applications, Vol.
  1, 3rd Edition, Revised Printing}.
\newblock Wiley, 1968.

\bibitem[GR09]{GR09}
R.~Guerraoui and E.~Ruppert.
\newblock Names trump malice: Tiny mobile agents can tolerate byzantine
  failures.
\newblock In {\em 36th International Colloquium on Automata, Languages and
  Programming (ICALP)}, volume 5556 of {\em LNCS}, pages 484--495.
  Springer-Verlag, 2009.

\bibitem[Jac05]{Ja05}
M.~O. Jackson.
\newblock A survey of network formation models: Stability and efficiency.
\newblock {\em Group Formation in Economics: Networks, Clubs and Coalitions,
  ed. G. Demange and M. Wooders}, pages 11--57, 2005.

\bibitem[Lyn96]{Ly96}
N.~A. Lynch.
\newblock {\em Distributed Algorithms}.
\newblock Morgan Kaufmann; 1st edition, 1996.

\bibitem[MCS11a]{MCS11-2}
O.~Michail, I.~Chatzigiannakis, and P.~G. Spirakis.
\newblock Mediated population protocols.
\newblock {\em Theoretical Computer Science}, 412[22]:2434--2450, May 2011.

\bibitem[MCS11b]{MCS11}
O.~Michail, I.~Chatzigiannakis, and P.~G. Spirakis.
\newblock {\em New Models for Population Protocols}.
\newblock N. A. Lynch (Ed), Synthesis Lectures on Distributed Computing Theory.
  Morgan \& Claypool, 2011.

\bibitem[MR95]{MR95}
R.~Motwani and P.~Raghavan.
\newblock {\em Randomized algorithms}.
\newblock Cambridge university press, 1995.

\bibitem[RW00]{RW00}
P.~W.~K. Rothemund and E.~Winfree.
\newblock The program-size complexity of self-assembled squares.
\newblock In {\em Proceedings of the 32nd annual ACM symposium on Theory of
  computing (STOC)}, pages 459--468, 2000.

\bibitem[Sch11]{Sc11}
J.~L. Schiff.
\newblock {\em Cellular automata: a discrete view of the world}, volume~45.
\newblock Wiley-Interscience, 2011.

\bibitem[SY99]{SY99}
I.~Suzuki and M.~Yamashita.
\newblock Distributed anonymous mobile robots: Formation of geometric patterns.
\newblock {\em SIAM J. Comput.}, 28[4]:1347--1363, March 1999.

\bibitem[VMC{\etalchar{+}}12]{VMC12}
N.~Vaidya, M.~L. Manapat, I.~A. Chen, R.~Xulvi-Brunet, E.~J. Hayden, and
  N.~Lehman.
\newblock Spontaneous network formation among cooperative {RNA} replicators.
\newblock {\em Nature}, 491[7422]:72--77, 2012.

\bibitem[Win98]{Wi98}
E.~Winfree.
\newblock {\em Algorithmic Self-Assembly of DNA}.
\newblock PhD thesis, California Institute of Technology, June 1998.

\bibitem[Zha03]{Zh03}
S.~Zhang.
\newblock Fabrication of novel biomaterials through molecular self-assembly.
\newblock {\em Nature biotechnology}, 21[10]:1171--1178, 2003.

\end{thebibliography}

\end{document}